\documentclass[11pt,onecolumn]{IEEEtran}

\usepackage[usenames]{color}
\usepackage{times}
\usepackage{algorithmic}
\usepackage{amsmath}
\usepackage{amssymb}
\usepackage{amsthm}
\usepackage{dsfont}
\usepackage{color}
\usepackage{colortbl}
\usepackage{float}
\usepackage{cite}
\usepackage{verbatim}
\usepackage[pdftex]{graphicx}
\usepackage{caption}
\usepackage{subcaption}

\usepackage[table]{xcolor}

\newtheorem{prop}{Proposition}
\newtheorem{lem}[prop]{Lemma}
\newtheorem{thm}[prop]{Theorem}
\newtheorem{cor}[prop]{Corollary}

\theoremstyle{definition}
\newtheorem{defi}[prop]{Definition}

\newtheorem{ex}[prop]{Example}
\newtheorem{rem}[prop]{Remark}

\newcommand{\PG}{\mathcal{P}_q(n)}
\newcommand{\rk}{\mathrm{rank}}
\newcommand{\wt}{\mathrm{wt}}

\newcommand{\Gr}{\mathcal{G}_q(k,n)}

\newcommand{\field}[1]{\mathbb{#1}}
\newcommand{\C}{\field{C}}
\newcommand{\F}{\field{F}}

\newcommand{\cF}{{\cal F}}

\newcommand{\cA}{{\cal A}}

\newcommand{\cC}{{\cal C}}
\newcommand{\cG}{{\cal G}}

\newcommand{\cP}{{\cal P}}

\newcommand{\deff}{\mbox{$\stackrel{\rm def}{=}$}}

\newcommand{\Gauss}[2]{
\left[\begin{array}
{c}#1\\#2\end{array}\right]_{q}
}

\begin{document}

\title{Subspace Codes based on Graph Matchings,\\ Ferrers Diagrams and Pending Blocks}

\author{
Natalia Silberstein and Anna-Lena Trautmann
\thanks{N.~Silberstein is with the Department of Computer Science, Technion --- Israel Institute of Technology, Haifa 32000, Israel (email: natalys@cs.technion.ac.il).}
\thanks{A.-L.~Trautmann is with the Institute of Mathematics, University of Zurich, Switzerland (email: anna-lena.trautmann@math.uzh.ch).}

\thanks{The first author is supported in part at the Technion by a Fine Fellowship.}
\thanks{The second author was partially supported by Forschungskredit of the University of Zurich, grant no. 57104103, and Swiss National Science Foundation Fellowship no. 147304.}
\thanks{Parts of this work were presented at ISIT 2013 in Istanbul, Turkey.}
\thanks{Copyright (c) 2014 IEEE.
}

}

\maketitle

\begin{abstract}
  This paper provides new constructions and lower bounds for subspace codes, using 
  Ferrers diagram rank-metric codes from matchings of the complete graph and pending  blocks.
  We present different constructions for constant dimension codes with minimum injection distance $2$ or $k-1$, where $k$ is the constant dimension. Furthermore, we present a construction of new codes from old codes for any minimum distance. Then we construct non-constant dimension codes from these codes.
  Some examples of codes obtained by these constructions are the  largest known codes for the given parameters.
  \end{abstract}




\section{Introduction}
\label{sec:intro}

Let $\F_q$ be the finite field of size $q$. Given two integers $k,n$, such that $0\leq k\leq n$, the set of all $k$-dimensional subspaces of $\F_q^n$ forms the Grassmannian over $\F_q$, denoted by $\Gr$. It is well known that the cardinality of the Grassmannian is given by the $q$-\emph{ary Gaussian coefficient}
\[|\Gr|=\Gauss{n}{k} \deff\ \prod_{i=0}^{k-1}\frac{q^{n-i}-1}{q^{k-i}-1}.
\]

The set of all subspaces of $\F_q^n$ is denoted by $\PG$. It holds that $\PG=\bigcup_{k=0}^n \Gr$.

Both the  \emph{subspace distance}, defined as
\begin{equation}
\label{def_subspace_distance}
d_S (X,\!Y) \,\ \deff\ \dim X + \dim Y -2 \dim\bigl( X\, {\cap}Y\bigr),
\end{equation}
 and the \emph{injection distance}, defined as
 \begin{equation}
\label{def_injection_distance}
d_I (X,\!Y) \,\ \deff\ \max\{\dim X , \dim Y\} - \dim\bigl( X\, {\cap}Y\bigr),
\end{equation}
for any two distinct subspaces $X$ and $Y$ in $\PG$,
  are metrics on $\PG$, and hence also on $\Gr$. Note that for $X,Y\in \Gr$ it holds that $d_S(X,Y)=2 d_I(X,Y)$.

We say that $\C\subseteq \Gr$ is an $(n,M,d,k)_q$ \emph{code in
the Grassmannian}, or \emph{constant-dimension code}, if $M =
|\C|$ and
the minimum injection distance of the code is $\min\{d_I(X,Y) \mid X,Y \in \C, X\neq Y\}=d$.
Since $d_S(X,Y)=2 d_I(X,Y)$ for $X,Y\in \Gr$, the minimum subspace distance of a constant-dimension code is twice the minimum injection distance of the code, thus one can equivalently use the subspace distance instead of the injection distance. Both notations, with the injection or with the subspace minimum distance, can be found in the literature.
Furthermore, we call $\C\subseteq \PG$ an $(n,M,d)_q^S$ \emph{subspace code}, or \emph{projective space code}, if $M =
|\C|$ and the minimum subspace distance of the code is $\min\{d_S(X,Y) \mid X,Y \in \C, X\neq Y\}=d$.
If we use the injection distance instead of the subspace distance, we call  $\C\subseteq \PG$ an $(n,M,d)_q^I$ subspace code.
$A_q(n,d,k)$ will denote the maximum size of an $(n,M,d,k)_q$ code.
By $A^*_q(n,d,k)$ we denote the size of the largest known $(n,M,d,k)_q$ code.

Subspace codes, and constant dimension codes in particular, have drawn significant attention in the last six years due to
the work by Koetter and Kschischang~\cite{ko08}, where they
presented an application of such codes for error correction in
random network coding. Constructions and bounds for constant dimension codes were
given e.g.\ in~\cite{bo09,et09,et12,et11,si11Enum,si11Lexi,ga11,ga10,kh09,ko08p,ma08p,sk10,tr10}. For non-constant dimension codes some results can be found in~\cite{et09,kh09,kh09master,et11,kh09p}.

One notes that the codes obtained by a simple construction based on \emph{lifting} of \emph{maximum rank distance} (MRD) codes~\cite{si08j} are almost optimal, i.e., asymptotically attain the known upper bounds~\cite{et11,ko08}. However, it is of interest to provide constructions of constant dimension codes which are larger than the lifted MRD codes. The first step in this direction was done in~\cite{et09}, where the \emph{multilevel} construction was presented. This construction generalizes the lifted  MRD codes construction by introducing a new family of rank-metric codes having a given shape of their codewords, namely, \emph{Ferrers diagram rank-metric codes}. Further, some other constructions were presented in~\cite{bo09,ga11,ga10,et12,et11,kh09,sk10,tr10}. Most of them provide constant dimension codes which contain a lifted MRD code as a subcode. Another type of constructions includes \emph{orbit} or \emph{cyclic} codes~\cite{et11,ko08p,tr10p}. In~\cite{et12}, an upper bound on the cardinality of codes which contain a lifted MRD code was presented for some sets of parameters. For constant dimension $k=3$ this bound was attained by using a generalization of a \emph{pending dots} based construction, presented in~\cite{tr10}.

In this paper, we continue with this direction of constructing large constant dimension codes which contain lifted MRD codes. We present new families of codes which have the largest known cardinality. The ideas for these constructions generalize the ideas presented in~\cite{si08j,et09,tr10,et12}.

First, we present new $(n,M,k-1,k)_q$ codes. These codes have the second largest possible injection distance $k-1$  (codes  having the largest possible injection distance $k$ are called \emph{(partial) spread} codes and were considered in e.g.~\cite{be75,go13,ma08p}).
 This case corresponds to the largest error correction, where the code cardinality can be improved.
 Our new codes are based on a two-dimensional generalization of pending dots, which we call \emph{pending blocks}. Based on this approach we construct $(n,M,k-1,k)_q$ codes of cardinality
\begin{equation}
\label{eq:cardIntro}
M =  q^{2(n-k)}+\sum_{j=3}^{k-1} q^{2(n-\sum_{i=j}^k i)}+\Gauss{n-\frac{k^2+k-6}{2}}{2}.
\end{equation}
Note that our new construction requires the field size $q$ to be large enough, namely, $q^2+q+1\geq n-\frac{k^2+k-6}{2}$.
For smaller fields however, we slightly modify the construction and obtain codes that have almost the same cardinality as in~(\ref{eq:cardIntro}).

Next, we focus on codes with the smallest non-trivial injection distance $d_I=2$ (a code with the smallest possible distance $d_I=1$ is the trivial code which contains the whole Grassmannian). It was shown in~\cite{et12} that the gap between the cardinality of a lifted MRD code and the known upper bounds increases for smaller values of the minimum distance. Thus the  minimum distance $d_I=2$ corresponds to the case where the most significant improvement in the code cardinality is possible.
We start with the multilevel construction of~\cite{et09}. The main drawback of this construction is that it depends on the choice of the underlying constant weight code, but the best choice for such a code is still unknown. As a consequence, the cardinality of constant dimension codes obtained by the multilevel construction can not be written in a general form. We consider a specific choice of a constant weight code for the multilevel construction. This constant weight code is based on an one-factorization of a complete graph.
The cardinality of the proposed $(n,M,2,k)_q$ code can be derived and recursively gives the lower bound
\begin{eqnarray}
\label{eq:card_B}
&A_q(n,2,k) \; \geq \hspace{0.4cm} \sum_{i=1}^{\lfloor\frac{n-2}{k}\rfloor -1} \Big( q^{(k-1)(n-ik)}+\nonumber \left.\frac{(q^{2(k-2)}-1)(q^{2(n-ik-1)}-1)}{(q^4-1)^2} q^{(k-3)(n-ik-2)+4}\right).
\end{eqnarray}

Then, we combine the idea of one-factorization based constant weight codes with the pending blocks construction and present a new family of $(n,M,2,k)_q$  codes.  Here, we use the one-factorization of a specific node labelling of the complete graph to provide codes with large cardinality.

In addition,  we present a simple way to construct a new constant dimension code from an old one, with the same minimum distance. Surprisingly, for some parameters this construction provides the largest known codes  (see Table II in Section~\ref{sec:Tables}).
In particular, we derive the following recursive formula for the maximum cardinality of a constant dimension code, for any $n\geq 3k$ and
$n\geq \Delta  \geq k$:
\[A_q(n,d,k) \geq q^{\Delta (k-d+1)} A_q(n-\Delta ,d,k) + A_q(\Delta,d,k)  .\]

We compare our constructions with other known constructions of constant dimension codes. For this we first analyze the difference of the cardinalities of our first three constructions with the cardinality formula of the \emph{multicomponent construction} \cite{ga11,tr13phd}, which is the largest known general construction with a closed cardinality formula. We show that the improvement of our construction compared to the multicomponent construction grows exponentially in the dimension $n$ of the ambient space, for the relevant cases with $d=2$ or $d=k-1$.
Next we compare our constructions for some parameter sets with the multilevel construction.
This construction does not have a closed cardinality formula, but gives rise to the largest known constant dimension codes for many parameter sets. One can see that in some cases our constructions beat the multilevel construction, while in other cases they do not (see Tables I and II in Section~\ref{sec:Tables}).

Finally, we consider non-constant dimension codes.
We use the constant dimension codes constructed in this paper as well as the largest codes from~\cite{et09,et12} and apply the puncturing method~\cite{et09} to obtain large codes for both the subspace and the injection metric.


The rest of this paper is organized as follows.
In Section~\ref{sec:preliminaries} we introduce the necessary definitions and two known constructions which will be the starting points to our new constructions. In Section~\ref{sec:recursion} we introduce the notation of pending blocks and present a construction for an $(n, M, k-1,k)_q$ code.
In Section~\ref{sec:constructionsk=4} we consider properties of Ferrers diagrams  arising from matchings of complete graphs and discuss the constructions for $(n, M, 2,k)_q$ codes. In Section~\ref{sec:new-old} we present a construction of a new code from a given one. Section~\ref{sec:Tables} presents the comparison between the new codes obtained in the paper and some previously known codes. We consider constructions of non-constant dimension codes in Section~\ref{sec:non-constant} and conclude with Section~\ref{sec:conclusion}.


\section{Preliminaries and Related Work}
\label{sec:preliminaries}

 In this section we briefly provide the definitions and previous results used in our constructions. More details can be found in~\cite{et09,et12,tr10}.

\subsection{Representations of Subspaces and Multilevel Construction}

Let $X$ be a $k$-dimensional subspace of $\F_q^n$. We represent $X$ by the matrix $\mbox{RE}(X)$ in reduced
row echelon form, such that the rows of $\mbox{RE}(X)$ form a basis of $X$. The \emph{identifying
vector} of $X$, denoted by $v(X)$, is the binary vector
of length $n$ and weight $k$, where the  $k$ \emph{ones} of $v(X)$ are exactly in the
positions where $\mbox{RE}(X)$ has the leading coefficients (the
pivots).
All the binary vectors of length $n$ and weight $k$
can be considered as  the identifying vectors of all the subspaces
in $\Gr$.
These $\binom{n}{k}$ vectors partition  $\Gr$
into $\binom{n}{k}$ different classes, where each class, also called a \emph{cell} of $\Gr$,
consists of all subspaces in $\Gr$ with the same identifying
vector.

Recall that the Hamming metric on $\F_q^n$ is defined as $d_H(u,v)\deff \wt(u-v)$, where $\wt(w)$ denotes the number of nonzero entries in the vector $w$. The asymmetric metric on $\F_2^n$ is defined as $d_{asym}(u,v)\deff \max\{N(u,v), N(v,u)\}$, where $N(u,v)$ denotes the number of coordinates $i$ where $u_i=1$ and $v_i=0$ \cite{kh09}.
The following results are useful tools for constructions of subspace codes.

\begin{prop}[\cite{et09,kh09,kh09p}]\label{prop3}
For
$X,Y\in\PG$ we have
\begin{itemize}
  \item $d_S(X,Y)\geq d_H(v(X),v(Y))$  ,
  \item $d_I(X,Y) \geq d_{asym}(v(X),v(Y))$  .
\end{itemize}
\end{prop}

The {\it Ferrers tableaux form} of a subspace $X$, denoted by
$\cF(X)$, is obtained from $\mbox{RE}(X)$ first by removing from
each row of $\mbox{RE}(X)$ the {\it zeros} to the left of the
leading coefficient; and after that removing the columns which
contain the leading coefficients. All the remaining entries are
shifted to the right. The \emph{Ferrers diagram} of $X$, denoted
by $\cF_X$, is obtained from $\cF(X)$ by replacing the entries of
$\cF(X)$ with dots.

 Given $\cF(X)$, the unique corresponding
subspace ${X\in\Gr}$ can easily be found. Also given $v(X)$, the unique corresponding $\cF_X$ can be found. When we fill the dots of a Ferrers diagram by elements of $\F_q$, we obtain a $\cF(X)$ for some $X\in \Gr$.

\begin{ex}  Let $X$ be the subspace in $\mathcal G_2(3,7)$ with
the  following  generator matrix in reduced row echelon form:

\begin{footnotesize}
$$\mbox{RE}(X)=\left( \begin{array}{ccccccc}
\bf{1} & \color{red}0 & 0 & 0 & \color{red} 1  & \color{red} 1& \color{red} 0\\
0 & 0 & \bf{1} & 0 & \color{red} 1 & \color{red}0 & \color{red}1 \\
0 & 0 & 0 &  \bf{1} &\color{red} 0& \color{red} 1 & \color{red} 1
\end{array}
\right) ~.$$
\end{footnotesize}
Its identifying vector is $v(X)=1011000$, and its Ferrers tableaux
form and Ferrers diagram are given by

\begin{footnotesize}
$$\begin{array}{cccc}
0 & 1 & 1 & 0 \\
&1 & 0 & 1  \\
&0 & 1 & 1
\end{array},~~~
\;
\begin{array}{cccc}
 \bullet & \bullet & \bullet & \bullet \\
  & \bullet & \bullet & \bullet   \\
  & \bullet & \bullet & \bullet  \\
\end{array},$$
\end{footnotesize}
respectively.
\end{ex}

In the following we will consider Ferrers diagram rank-metric codes, which are closely related to constant dimension codes.
For two $m \times \ell$ matrices $A$ and $B$ over $\F_q$ the {\it
rank distance}, $d_R(A,B)$, is defined by
$$
d_R (A,B)  \deff\text{rank}(A-B).
$$
\begin{prop}[\cite{et09,kh09}]
\label{pr:equal id}
For
$X,Y\in\PG$ we have that
 if $v(X)=v(Y)$ then
 \begin{itemize}
   \item $d_S(X,Y)=2d_R(\mbox{RE}(X),\mbox{RE}(Y))$,
   \item $d_I(X,Y)=d_R(\mbox{RE}(X),\mbox{RE}(Y))$.
 \end{itemize}
\end{prop}

Let $\cF$ be a Ferrers diagram with $m$ dots in the rightmost
column and $\ell$ dots in the top row. A code $\cC_{\cF}$ is an
$[\cF,\rho,d ]$ {\it Ferrers diagram rank-metric (FDRM) code}  if
all codewords of $\cC_{\cF}$ are $m\times \ell$ matrices in which
all entries not in $\cF$ are {\it zeros}, they form a linear subspace of dimension $\rho$ of
$\F_q^{m \times \ell}$, and for any two distinct codewords $A$
and $B$, $d_R (A,B) \geq d $.
If $\cF$ is a rectangular $m\times \ell$ diagram with $m \ell$ dots then the FDRM code is a classical rank-metric code~\cite{ga85a,ro91}.
The following theorem provides an upper bound on the cardinality of $\cC_{\cF}$.

\begin{thm}[\cite{et09}]\label{thm1}
Let $\cF$ be a Ferrers diagram and $\cC_{\cF}$ the corresponding $[\cF,\rho,d ]$ FDRM code.
Then
$ \rho \leq \min_i\{w_i\}$,
where $w_i$ is the number of dots in $\cF$ which are not contained in the first
$i$ rows and the rightmost $d -1-i$ columns ($0\leq i\leq d -1$).
\end{thm}

A code which attains the bound of Theorem~\ref{thm1} is
called a \emph{Ferrers diagram maximum rank distance (FDMRD) code}.
%
Maximum rank distance (MRD) codes are a class of $[\cF,\ell(m-d +1),d ]$ FDMRD codes, $\ell\geq m$, with a full $m\times \ell$ diagram $\cF$, which attain the bound of Theorem~\ref{thm1}~\cite{ga85a,ro91}.

It was proved in~\cite{et09} that for general diagrams the bound of Theorem~\ref{thm1} is attained for $d =1,2$:
\begin{thm}\label{thm2}
For any Ferrers diagram $\cF$ there exists an $[\cF,\rho,d ]$ FDMRD code for $d=1$ or $d=2$.
\end{thm}
Some special cases, when this bound is attained for $d >2$, can also be found in~\cite{et09}.

For a codeword $A \in \cC_{\cF} \subseteq \F_q^{k \times (n-k)}$
let $A_{\cF}$ denote the part of
$A$ related to the entries of $\cF$ in $A$.

\begin{defi}
Given an
 FDMRD code $\cC_{\cF}$, a \emph{lifted FDMRD code}
$\C_{\cF}$ is defined as follows:
$$\C_{\cF} = \{X\in \Gr :
\cF(X)=A_{\cF},~ A \in \cC_{\cF} \}.
$$
\end{defi}
This definition is the generalization of the definition of a lifted
MRD code~\cite{si08j}. Note, that all the codewords of a lifted MRD code have the same identifying vector of the type $(\underset{k}{\underbrace{11...1}}\underset{n-k}{\underbrace{000...00}})$.
The following theorem~\cite{et09} is the generalization of the result
given in~\cite{si08j}.

\begin{thm}
\label{lem:dist_lift} If $\cC_{\cF} \subset \F_q^{k \times (n-k)}$
is an $[ \cF , \rho , d  ]$ Ferrers diagram
rank-metric code, then its lifted code $\C_{\cF}$ is an $(n, q^\rho ,
d  , k)_q$ constant dimension code.
\end{thm}


 The multilevel construction~\cite{et09} for constant dimension codes is based on Proposition~\ref{prop3} and Theorem~\ref{lem:dist_lift}:

{\textbf{Multilevel Construction.}}
 First,  a binary constant weight code of
length $n$, weight $k$, and Hamming distance $2 d$ is
chosen to be the set of the identifying vectors
for~$\C$. Then, for each identifying vector a corresponding lifted
 FDRM code with minimum injection distance $d$
is constructed. The union of these lifted 
 FDRM
codes is an $(n, M, d , k)_q$ code.


\subsection{One-Factorization of Complete Graphs and the Pending Dots Construction}
\label{subsec:Pending Dots Preliminaries}
In the construction provided in~\cite{et12}, for $k=3$ and $d =2$, in the stage of choosing the identifying vectors for a code $\C$, a set of vectors with minimum (Hamming) distance $2d -2=2$ is allowed, by using a method based on pending dots in a Ferrers
diagram~\cite{tr10}, which will be explained in the following.


The \emph{pending dots} of a Ferrers diagram $\cF$ are the leftmost dots in the
first row of $\cF$ whose removal has no impact on the size of the
corresponding Ferrers diagram rank-metric code. The following lemma follows from~\cite{tr10}.

\begin{lem}\label{lm:pending dots}
Let $X$ and $Y$  be two subspaces in $\Gr$ with
$d_H(v(X),v(Y))=2d -2$, such that the leftmost \emph{one} of
$v(X)$ is in the same position as the leftmost  \emph{one} of
$v(Y)$.  Let $P_X$ and $P_Y$ be the sets of pending dots of $X$ and $Y$, respectively.
If $P_X\cap P_Y\neq \varnothing$
and the entries in $P_X\cap P_Y$ (of their Ferrers tableaux forms) are
assigned with different values in at least one position, then
$d_S(X,Y) = 2d_I(X,Y) \geq 2d .$
\end{lem}

\begin{ex}
Let $X$ and $Y$  be subspaces in $\mathcal{G}_q(3,6)$ which are given by
the following generator matrices:
\begin{footnotesize}
$$\left(\begin{array}{cccccc}
1 &\textcircled{\raisebox{-0.9pt}{0}} &0  & v_1 & v_2 &  0\\
0 & 0 & 1  & v_3 & v_4 &  0 \\
0 & 0 & 0 &  0 &  0 & 1
\end{array}\right),\:
\left(\begin{array}{cccccc}
1 &\textcircled{\raisebox{-0.9pt}{1}} &u_1 & 0 &  u_2  & 0\\
0 & 0 & 0  & 1  & u_3 & 0  \\
0 & 0 & 0 & 0 &  0  & 1
\end{array}\right)
$$
\end{footnotesize}
where $v_i ,u_i\in \F_q$, and the pending dots are emphasized by
circles. Their identifying vectors are $v(X)=101001$ and
$v(Y)=100101$. Clearly, $d_H(v(X), v(Y))=2$, while $d_S(X,Y)\geq 4$.
\end{ex}

The following results from the area of graph theory will be useful in the following code constructions.
We denote by $K_m$ the \emph{complete graph} with $m$ nodes. A \emph{matching} of $K_m$ is a set of non-adjacent edges of $K_m$. A \emph{perfect (resp. nearly perfect) matching} is a matching that covers all (resp. all but one) nodes of $K_m$.
A \emph{one-factorization (OF)} (resp. \emph{near one-factorization (NOF)}) of $K_m$ is a partition of all edges into perfect (resp. nearly perfect) matchings of $K_m$.
If one labels all nodes of $K_m$ with the numbers from $1,\dots, m$, then one can easily see the  $1-1$-correspondence between the edges of the graph and the weight-$2$ vectors of $\F_2^m$ by assigning the two ones of the vector in the coordinates labelled by the numbers of the two nodes in the graph which are connected by the corresponding edge.

The following lemma, which  follows from a one-factorization and near-one-factorization of  a complete graph~\cite{tu84,li01b}, will be used in our  constructions.

\begin{lem}\label{lm:1-factorization}
Let $D$ be the set of all binary vectors  of length~$m$ and weight
$2$.
\begin{itemize}
\item If $m$ is even, $D$ can be partitioned into $m-1$ classes,
each of $\frac{m}{2}$ vectors with pairwise disjoint positions of \emph{ones};
\item If $m$ is odd, $D$ can be partitioned into $m$ classes,
each of $\frac{m-1}{2}$ vectors with pairwise disjoint positions of \emph{ones}.
\end{itemize}
\end{lem}

The following construction for $k=3$ and $d =2$ based on pending dots from~\cite{et12} will be used as the base step of our recursive construction proposed in the sequel.

{\textbf{Pending Dots Construction.}}\label{sec:const0}
 Let $n\geq 8$ and $q^2+q+1\geq \ell$, where $\ell=n-4$ for odd $n$ and $\ell=n-3$ for even~$n$.
 In addition to the lifted MRD code (which has the identifying vector $v_0=(11100\ldots0)$), the final code $\C$ will contain the codewords with identifying vectors of the form $(x||y)$, where the prefix $x\in \F_2^3$ is of weight $1$ and the suffix $y\in \F_2^{n-3}$ is of weight~$2$. By Lemma~\ref{lm:1-factorization}, we partition the set of suffixes into $\ell$  classes $P_1,P_2,\ldots,P_{\ell}$ and define the following three sets:
\[\cA_1=\{(001||y):y\in P_1\},\]
\[\cA_2=\{(010||y):y\in P_i, 2\leq i\leq \min\{q+1,\ell\}\},\]
\[\cA_3= \left\{\begin{array}{cc}
                \{(100||y):y\in P_i,~ q+2\leq i\leq \ell\} & \textmd{if }\ell>q+1 \\
                \varnothing & \textmd{if }\ell\leq q+1 \\
              \end{array}\right..
\]
Elements with the same prefix and distinct suffixes from the same class $P_{i}$ have Hamming distance~$4$. When we use the same prefix for two different classes $P_i, P_j$, we assign different values in the pending dots of the Ferrers tableaux forms.
Then the corresponding lifted FDMRD codes of injection distance~$2$ are constructed, and their union with the lifted MRD code forms the final code $\C$ of size $q^{2(n-3)}+\Gauss{n-3}{2}$.

In the following sections we will generalize this construction in various ways and obtain codes for any $k\geq4$ with minimum injection distance $d =2$ or with $d =k-1$, or equivalently minimum subspace distance $2d =4$ or with $2d =2(k-1)$.


\section{Construction for $(n,M,k-1,k)_{q}$ Codes}
\label{sec:recursion}

In this section we provide a recursive construction for  $(n,M,k-1,k)_{q}$  codes, which uses the Pending Dots construction 
described in Section~\ref{sec:preliminaries} as an initial step. Codes obtained by this construction contain a lifted MRD code. The upper bound on the cardinality of such codes is derived  in~\cite{et12} and given in the following theorem.

\begin{thm}[\cite{et12}]
\label{trm:upper bound from Steiner Structure}
If an $(n,M,k-1,k)_q$ code $\mathbb{C}$, $k \geq 3$, contains
an $(n,q^{2(n-k)},k-1,k)_q$ lifted MRD code then
$$M\leq
q^{2(n-k)}+ A_q (n-k,k-2,k-1).$$
\end{thm}

Note that for $k=3$ this bound is given by
$$
M\leq q^{2(n-3)}+\Gauss{n-3}{2},
$$
which is attained by the Pending Dots construction.
Our recursive construction provides a new lower bound on the cardinality of such codes for general $k$.
%
%

To present the construction
we first need to extend the definition of pending dots of~\cite{tr10} to a two-dimensional setting, which we will do in the following subsection.

\vspace{0.3cm}

\subsection{Pending Blocks}
\label{sec:pendingBlocks}

\begin{defi}
Let $\cF$ be a Ferrers diagram with $m$ dots in the rightmost column and $\ell$ dots in the top row. We say that the $\ell_1<\ell$ leftmost columns of $\cF$ form a \emph{pending block} (of length $\ell_1$) if the upper bound on the size of FDMRD code $\cC_{\cF}$  from Theorem~\ref{thm1} is equal to the upper bound on the size of $\cC_{\cF}$ without the $\ell_1$ leftmost columns.
\end{defi}

\begin{ex}\label{ex:pending_block}
Consider the following Ferrers diagrams:
\begin{footnotesize}
\[ \cF_1=\begin{array}{cccccc} \bullet & \bullet &\bullet&\bullet&\bullet \\
 &\bullet& \bullet & \bullet &\bullet\\
& &\bullet &\bullet&\bullet
\end{array}\quad, \quad
\cF_2=\begin{array}{cccccc}\bullet&\bullet&\bullet \\
  \bullet & \bullet &\bullet\\
\bullet &\bullet&\bullet
\end{array} .\]
\end{footnotesize}

For $d =3$ by Theorem~\ref{thm1} both codes $\cC_{\cF_1}$ and $\cC_{\cF_2}$ have $|\cC_{\cF_i}|\leq q^3$, $i=1,2$.
The diagram $\cF_1$ has the pending block $ \begin{array}{cccccc} \bullet & \bullet\\&\bullet
\end{array}$ and the diagram $\cF_2$ has no pending block.
\end{ex}

\begin{defi}
Let $\cF$ be a Ferrers diagram with $m$ dots in the rightmost column and $\ell$ dots in the top row, and let $\ell_1< \ell$, and $m_1< m$. If the $(m_1+1)$st row of $\cF$ has less dots than the $m_1$th row of $\cF$ and at most $m-\ell_1$ dots, then the $\ell_1$ leftmost columns of $\cF$ are called a \emph{quasi-pending block} (of size $m_1\times\ell_1$).
\end{defi}

Note that a pending block is also a quasi-pending block.
We can now generalize Lemma \ref{lm:pending dots} from pending dots to pending blocks.

\begin{thm}\label{thm4}
Let $X,Y\in \Gr$, such that $\rm RE(X)$ and $\rm RE(Y)$ have a quasi-pending block of size $m_{1}\times \ell_{1}$ in the same position and $d_{H}(v(X),v(Y))=  2d$.   Denote the submatrices of $\cF(X)$ and $\cF(Y)$ corresponding to the quasi-pending blocks by $B_X$ and $B_Y$, respectively.
Then
$d_{I}(X,Y) \geq  d+\rk(B_X- B_Y) $ or equivalently
$d_{S}(X,Y) \geq 2 d+2\rk(B_X- B_Y) $.
\end{thm}

\begin{proof}
Since the quasi-pending blocks are in the same position, the first $h$ pivots of $\mbox{RE}(X)$ and $\mbox{RE}(Y)$ are in the same columns.
To compute the rank of $\left[ \begin{array}{cc}\mbox{RE}(X)\\\mbox{RE}(Y)\end{array} \right] $ we permute the columns such that the $h$ first pivot columns are to the very left, then the columns of the pending block, then the other pivot columns and then the rest:

\[\rk\left[ \begin{array}{cc}\mbox{RE}(X)\\\mbox{RE}(Y)\end{array} \right]\]\[
=
\rk
\begin{footnotesize}\left[ \begin{array}{cccccccccccccc}
1&\dots&0 &\cellcolor{gray}  &\cellcolor{gray}&\cellcolor{gray} &\cellcolor{gray}&\cellcolor{gray} &0&\dots\\
\vdots&\ddots &\vdots&&&\cellcolor{gray}\ddots&\cellcolor{gray}B_X&\cellcolor{gray}\vdots&\vdots&\vdots\\
0&\dots&1&0&\dots&0&\cellcolor{gray} &\cellcolor{gray} &0&\dots\\
0&\dots&0&0&\hdots&0&\hdots&0&1&\hdots\\
\vdots&&&&&&&&\vdots\\
\hline
1&\dots&0 &\cellcolor{gray}  &\cellcolor{gray}&\cellcolor{gray} &\cellcolor{gray}&\cellcolor{gray} &0&\dots\\
\vdots&\ddots &\vdots&&&\cellcolor{gray}\ddots&\cellcolor{gray}B_Y&\cellcolor{gray}\vdots&\vdots&\vdots\\
0&\dots&1&0&\dots&0&\cellcolor{gray} &\cellcolor{gray} &0&\dots\\
0&\dots&0&0&\hdots&0&\hdots&0&1&\hdots\\
\vdots&&&&&&&&\vdots
\end{array}\right]
\end{footnotesize}
.\]
Now we subtract the lower half of the obtained matrix from the upper half and write the result in the lower half of the new matrix to get
\[=\rk
\begin{footnotesize}\left[ \begin{array}{cccccccccccccc}
1&\dots&0 &\cellcolor{gray}  &\cellcolor{gray}&\cellcolor{gray} &\cellcolor{gray}&\cellcolor{gray} &0&\dots\\
\vdots&\ddots &\vdots&&&\cellcolor{gray}\ddots&\cellcolor{gray}B_X&\cellcolor{gray}\vdots&\vdots&\vdots\\
0&\dots&1&0&\dots&0&\cellcolor{gray} &\cellcolor{gray} &0&\dots\\
0&\dots&0&0&\hdots&0&\hdots&0&1&\hdots\\
\vdots&&&&&&&&\vdots\\
\hline
0&\dots&0 &\cellcolor{gray}  &\cellcolor{gray}&\cellcolor{gray} &\cellcolor{gray}&\cellcolor{gray} &0&\dots\\
\vdots&\ddots &\vdots&&&\cellcolor{gray}\ddots&\cellcolor{gray}B_X-B_Y&\cellcolor{gray}\vdots&\vdots&\vdots\\
0&\dots&0&0&\dots&0&\cellcolor{gray} &\cellcolor{gray} &0&\dots\\
0&\dots&0&0&\hdots&0&\hdots&0&0&\hdots\\
\vdots&&&&&&&&\vdots
\end{array}\right]
\end{footnotesize}   .\]
The additional pivots of $\mbox{RE}(X)$ and $\mbox{RE}(Y)$ (to the right in the above representation) that were in different columns in the beginning are still in different columns, hence it follows that
$\rk\left[ \begin{array}{cc}\rm RE(X)\\\rm RE(Y)\end{array} \right] \geq  k+\rk(B_X-B_Y) +d $,
which implies the statement.
\end{proof}
This theorem implies that for the construction of an $(n,M,d ,k)_q$ code, by filling the (quasi-) pending blocks with a suitable Ferrers diagram rank metric code, one can choose a set of identifying vectors with lower minimum Hamming distance than $2d $.




\subsection{The Construction.}

The following two lemmas will be useful for our construction.
\begin{lem}\label{lem14}
\label{lm:gen_k-2}
 Let $n-k-2\geq n_1\geq k-2$ and $v$ be an identifying vector of length $n$ and weight $k$, such that there are $k-2$ many ones in the first $n_1$ positions of $v$. Then the Ferrers diagram arising from $v$ has more or equally many dots in any of the first $k-2$ rows  than in the last column,
 and the upper bound for the dimension of a Ferrers diagram code with minimum rank distance $k-1$ is the number of dots in the lower two rows.
\end{lem}
 \begin{proof}
 Naturally, the last column of the Ferrers diagram has at most $k$ many dots. It holds that any column has at most as many dots as the last one.
 Since there are $k-2$ many ones in the first $n_{1}$ positions of $v$, it follows that there are $n-n_{1}-2$ zeros in the last $n-n_{1}$ positions of $v$. Thus, there are at least $n-n_{1}-2$ many dots in any but the lower two rows of the Ferrers diagram arising from $v$. Therefore, if $n-n_{1}-2\geq k \iff n-k-2\geq n_{1}$ the Ferrers diagram arising from $v$ has more than or equally many dots in any of the first $k-2$ rows  than in the last column, and hence than in any column.

 From Theorem \ref{thm1} we know that the bound on the dimension of the FDRM code is given by the minimum number of dots not contained in the first $i$ rows and last $k-2-i$ columns for $i=0,\dots,k-2$. If we start with $i=k-2$ we get that the dimension of the code is at most the number of dots in the last two rows of the diagram. Inductively, if we decrease $i$ by one, we add a row (of the first $k-2$ rows) and erase a column of the previous diagram, which results in more points, hence  the minimum is attained for $i=k-2$.
 \end{proof}

\begin{rem}\label{rem5}
If an $m\times \ell$-Ferrers diagram has $d-1$  rows with $\ell$ dots each, then the construction of~\cite{et09} provides respective FDMRD codes of minimum distance $d$
attaining the bound of Theorem~\ref{thm1}.
\end{rem}


We need yet another special case of Ferrers diagrams where we can attain the upper bound on the dimension of the code size.

\begin{lem}\label{lemk-3}
For an $m\times \ell$-Ferrers diagram where the $j$th row has at least $x$ more dots than the $(j+1)$th row for $1\leq j \leq m-1$ and the lowest row has $x$ many dots, there is a FDMRD code with minimum rank distance $m$ and cardinality~$q^{x}$.
\end{lem}

\begin{proof} The construction is as follows: For each codeword take a different $w\in \F_{q}^{x}$ and fill the first $x$ dots of every row with this vector, whereas all other dots are filled with zeros.
The minimum distance follows easily from the fact that the positions of the $w$'s in each row have no column-wise intersection. Since they are all different, any difference of two codewords has a non-zero entry in each row and it is already row-reduced.

The cardinality is clear, hence it remains to show that this attains the bound of Theorem \ref{thm1}. Plugging in $i=k-1$ in Theorem \ref{thm1} we get that the dimension of the code is less than or equal to the number of dots in the last row, which is achieved by this construction.
\end{proof}

We now have all the machinery to describe the new construction for $(n,M,k-1,k)_q$ codes.

\noindent\textbf{Construction A.}
\\
Let $k\geq 4$,
$s := \sum_{i=3}^{k}i=\frac{k^2+k-6}{2}$,
$n\geq s+2+k=\frac{k^2+3k-2}{2}$ and
$q^2+q+1\geq \ell$, where $\ell:=n-s=n-\frac{k^2+k-6}{2}$ for odd $n-s$ (or $\ell:= n-s-1=n-\frac{k^2+k-4}{2}$ for even $n-s$).
 \noindent


\emph{ Identifying vectors:}
 In addition to the identifying vector $v_{00}^k=(11\ldots1100\ldots0)$ of the lifted MRD code $\C^k_*$ (of size $q^{2(n-k)}$ and minimum subspace distance $2(k-1)$), the other identifying vectors of the codewords are defined as follows.
First, by Lemma~\ref{lm:1-factorization}, we partition the weight-$2$ vectors of $\F_{2}^{n-s}$ into classes $P_{1},\dots, P_{\ell}$ of size $\frac{\bar \ell}{2}$  (where $\ell=\bar \ell -1=n-s-1$ if $n-s$ even and $\ell=\bar \ell +1=n-s$ if $n-s$ odd) with pairwise disjoint positions of the ones.
 We define the sets of identifying vectors  by a recursion. Let $v_0\in \F_q^{n-s+3}$ and $\cA_1,\cA_2,\cA_3\subseteq\F_q^{n-s+3}$,  as defined in the Pending Dots construction (see Section~\ref{subsec:Pending Dots Preliminaries}). Then $v_{00}^3=v_0$,
\[\cA_0^3= \emptyset\textmd{, } \cA_i^3=\cA_i\textmd{, } 1\leq i\leq 3.
\]
For $k\geq 4$ we define:
 \[\cA_0^k=\{v^k_{01},\dots,v^k_{0k-3}\},\]
where $v_{0j}^k=(000\; w^k_j \;|| v_{0j-1}^{k-1})$ ($1\leq j\leq k-3$),
 such that the $w_j^k$ are all different weight-$1$ vectors of $\F_2^{k-3}$.
Furthermore we define:
\begin{align*}
\cA_1^k&=\{(0010\dots 00 || z)  : z \in \cA_1^{k-1}\},\\
\cA_2^k&=\{(0100\dots 00 || z)  : z \in \cA_2^{k-1}\},\\
\cA_{3}^k&=\{(1000\dots00 || z )  : z \in \cA_3^{k-1}\},
\end{align*}
such that the prefixes of the vectors in $\bigcup_{i=0}^{3}\cA_i^k$ are  vectors of $\F_2^{k}$ of weight $1$.
Note, that the suffix $y\in\F_q^{n-s}$ (from the Pending Dots construction)
in all the vectors from $\cA_1^k$ belongs to $P_1$, the suffix $y$ in all the vectors from $\cA_2^k$ belongs to $\bigcup_{i=2}^{\min\{q+1,\ell\}} P_i$, and the suffix $y$ in all the vectors from $\cA_3^k$ belongs to $\bigcup_{i=q+2}^{\ell} P_i$ (the set $\cA_3^k$ is empty if $\ell\leq q+1$).

\vspace{.3cm}
\emph{Pending blocks:}
\begin{itemize}
  \item All Ferrers diagrams that correspond to the vectors  in $\cA_1^k$ have a common pending block with $k-3$ rows and 
  $ \sum_{i=3}^{k-j}i$ dots in the 
  $j$th row, for $1\leq j\leq k-3$. We fill each of these pending blocks with a different element of a suitable FDMRD code with minimum rank distance $k-3$ and size $q^3$, according to Lemma \ref{lemk-3}.
  Note, that the initial conditions always imply that $q^{3}\geq \bar \ell $.
  \item All Ferrers diagrams that correspond to the vectors  in $\cA_2^k$ have a common pending block with $k-2$ rows and $\sum_{i=3}^{k-j}i+1$ dots in the 
$j$th row, $1\leq j\leq k-2$.
Every vector which has a suffix $y$ from the same $ P_i$ will have the same value $a_i\in \F_q$ in the first entry in each row of the common pending block, such that the vectors with suffixes from the different classes will have different values in these entries.  (This corresponds to a FDMRD code of distance $k-2$ and size $q$.)
  Given the filling of the first entries of every row, all the other entries of the pending blocks are filled by a FDMRD code with minimum distance $k-3$, according to Lemma \ref{lemk-3}.
  \item All Ferrers diagrams that correspond to the vectors  in $\cA_3^k$ have a common pending block with $k-2$ rows and  $\sum_{i=3}^{k-j}i +2$ dots in the $j$th row, ${1\leq j\leq k-2}$. The filling of these pending blocks is analogous to the previous case, but for the suffixes from the different $P_{i}$-classes we fix the first two entries in each row of a pending block.

\end{itemize}

\emph{Ferrers tableaux forms:}
On the dots corresponding to the last $n-s-2$ columns of the Ferrers diagrams for each vector $v_j$ in a given $\cA^k_i$, $0\leq i\leq 3$,  we construct a FDMRD code with minimum distance $k-1$ (according to Remark \ref{rem5}) and lift it to obtain $\C_{i,j}^k$. We define $\C_i^k=\bigcup_{j=1}^{|\cA^k_i|}\C_{i,j}^k$.

\emph{Code:}
The final code is defined as
$$\C^k=\bigcup_{i=0}^{3}\C_{i}^k\cup \C^k_*.$$

\begin{thm}
\label{trm:recursive_parameters}
The code $\C^k$ obtained by Construction A has minimum injection distance $k-1$ and cardinality $|\C^k|=q^{2(n-k)}+q^{2(n-(k+(k-1)))}+\ldots
+q^{2(n-\frac{k^2+k-6}{2})}+\Gauss{n-\frac{k^2+k-6}{2}}{2} $.
\end{thm}
\begin{proof}
We will first prove the cardinality by induction on $k$.
Observe that
the only  identifying vector that contributes additional codewords in $\C^k$ compared to $\C^{k-1}$ is $v_{00}^k$, since for all the other identifying vectors, the additional line of dots of the corresponding Ferrers diagrams does not increase the cardinality due to Lemma~\ref{lm:gen_k-2}, and thus $|\C^k|=|\C^{k-1}|+q^{2(n-k)}$ for any $k\geq 4$. Solving this recursively results in the above formula.

Next we prove that the minimum injection distance of $\C^k$ is $k-1$.
Let $X,Y\in \C^k$, $X\neq Y$. If $v(X)=v(Y)$, then by Proposition~\ref{pr:equal id}, $d_S(X,Y)\geq 2(k-1)$, i.e.\ $d_I(X,Y)\geq k-1$. Now we assume that $v(X)\neq v(Y)$. Note, that according to the definition of identifying vectors, $d_I(X,Y)\geq d_H(v(X,v(Y))/2=k-1$  for $(X,Y)\in \C_*^k\times \C_i^k$, $0\leq i\leq 3$, for $(X,Y)\in \C_0^k\times \C_0^k$, and for  $(X,Y)\in \C_i^k\times\C_j^k$, $i\neq j$. Now let $X,Y\in \C_i^k$, for some $1\leq i\leq 3$.
\begin{itemize}
  \item If the suffixes of $X$ and $Y$ of length $n-3$ belong to the same class $P_t$, then $d_H(v(X), v(Y))=4$ and $d_R(B_X,B_Y)=k-3$, for the submatrices $B_X, B_Y$ of $\cF(X), \cF(Y)$  corresponding to the common pending blocks. Then by Theorem~\ref{thm4}, $d_I(X, Y)\geq 2+(k-3)=k-1$.
  \item If the suffixes of $X$ and $Y$ of length $n-3$ belong to different classes $P_{t_1},P_{t_2}$, then $d_H(v(X), v(Y))=2$ and $d_R(B_X,B_Y)=k-2$, for the submatrices $B_X, B_Y$ of $\cF(X), \cF(Y)$  corresponding to the common pending blocks. Then by Theorem~\ref{thm4}, $d_I(X, Y)\geq 1+(k-2)=k-1$.
\end{itemize}
Hence, for any $X,Y\in\C^k$ it holds that $d_I(X,Y)\geq k-1$.
\end{proof}

\begin{cor}
Let $n\geq \frac{k^2+3k-2}{2}$ and $q^2+q+1\geq \ell$, where $\ell= n-\frac{k^2+k-6}{2}$ for odd $n-\frac{k^2+k-6}{2}$ (or $\ell= n-\frac{k^2+k-4}{2}$ for even $n-\frac{k^2+k-6}{2}$). Then
%
\[A_{q}(n, k-1, k) \geq \]
\[q^{2(n-k)}+\sum_{j=3}^{k-1} q^{2(n-\sum_{i=j}^k i)}+\Gauss{n-\frac{k^2+k-6}{2}}{2} .\]
\end{cor}


\begin{ex}
Let $k=5$, $d=4$, $n=19$, and $q=2$.  The code $\C^5$ obtained by Construction A has cardinality $2^{28}+2^{20}+2^{14}+\Gauss{7}{2}=2^{28}+1067627$ (the largest previously known code is of cardinality $2^{28}+1052778$~\cite{et09}).
We now illustrate the construction:

First, we partition the set of suffixes $y\in \F_2^7$ of weight $2$ into $7$ classes, $P_1,\ldots,P_7$ of size $3$ each. The identifying vectors of the code are partitioned as follows:
\begin{align*}
v_{00}^5={}\hspace{0.12cm}&(11111||0000||000||0000000)\\
\cA_0^5=\{
&(00001||1111||000||0000000),(00010||0001||111||0000000)\}\\
\cA_1^5=\{&(00100||0010||001||y):y\in P_1\}\\
\cA_2^5=\{&(01000||0100||010||y):y\in \{P_2,P_3\}\}\\
\cA_3^5=\{&(10000||1000||100||y):y\in \{P_4,P_5,P_6,P_7\}\}
\end{align*}

To demonstrate the idea of the construction we will consider only the set $\cA_2^5$. All the codewords corresponding to $\cA_2^5$ have the following common pending block $B$:

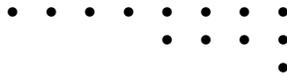
 \begin{footnotesize}
  \[\begin{array}{cccccccc}
  \bullet & \bullet & \bullet & \bullet  & \bullet & \bullet& \bullet  &  \bullet\\
   &&&& \bullet & \bullet & \bullet & \bullet \\
   &&& &&&& \bullet
        \end{array}
\]
\end{footnotesize}

If the suffix $y\in P_2$, or $y\in P_3$ then to distinguish between these two classes we assign the following values to $B$, respectively:

 \begin{footnotesize}
  \[\begin{array}{cccccccc}
  1 & \bullet & \bullet & \bullet  & \bullet & \bullet& \bullet  &  \bullet\\
   &&&& 1 & \bullet & \bullet & \bullet \\
   &&& &&&& 1
        \end{array},\textmd{ or }
        \begin{array}{cccccccc}
  0 & \bullet & \bullet & \bullet  & \bullet & \bullet& \bullet  &  \bullet\\
   &&&& 0 & \bullet & \bullet & \bullet \\
   &&& &&&& 0
        \end{array}
\]
\end{footnotesize}
For the  identifying vectors with the suffixes $y$ from $P_i$, $i=2,3$, we
construct a FDMRD code of distance $2$ for the remaining dots of $B$ (here, $a=0$ or $a=1$), as follows:
\[\begin{array}{cccccccc}
  $a$ & 0 & 0 & 0  & 0 & 0& 0  & 0\\
   &&&& $a$ &0 & 0 & 0 \\
   &&& &&&& $a$
        \end{array},
        \begin{array}{cccccccc}
   $a$ & 1 & 0 & 0  & 0 & 0& 0  & 0\\
   &&&& $a$ &1 & 0 & 0 \\
   &&& &&&& $a$
        \end{array},
        \]
        \[
        \begin{array}{cccccccc}
   $a$ & 0 & 1 & 0  & 0 & 0& 0  & 0\\
   &&&& $a$ &0 & 1 & 0 \\
   &&& &&&& $a$
        \end{array},
        \begin{array}{cccccccc}
   $a$ & 0 & 0 & 1  & 0 & 0& 0  & 0\\
   &&&& $a$ &0 & 0 & 1 \\
   &&& &&&& $a$
        \end{array},
        \]
\[\begin{array}{cccccccc}
  $a$ & 1 & 1 & 0  & 0 & 0& 0  & 0\\
   &&&& $a$ &1 & 1 & 0 \\
   &&& &&&& $a$
        \end{array},
        \begin{array}{cccccccc}
   $a$ & 1 & 0 & 1  & 0 & 0& 0  & 0\\
   &&&& $a$ &1 & 0 & 1 \\
   &&& &&&& $a$
        \end{array},
        \]
        \[
        \begin{array}{cccccccc}
   $a$ & 0 & 1 & 1  & 0 & 0& 0  & 0\\
   &&&& $a$ &0 & 1 & 1 \\
   &&& &&&& $a$
        \end{array},
        \begin{array}{cccccccc}
   $a$ & 1 & 1 & 1  & 0 & 0& 0  & 0\\
   &&&& $a$ &1 & 1 & 1 \\
   &&& &&&& $a$
        \end{array}.
        \]
      Since $P_{i}$ contains only three elements, we only need to use three of the above tableaux.
 We proceed analogously for the pending blocks of $\cA_{1}^{5}, \cA_{3}^{5}$. Then we fill the Ferrers diagrams corresponding to the last $7$ columns of the identifying vectors with an FDMRD code of minimum rank distance $4$ and lift these elements. Moreover, we add the lifted MRD code corresponding to $v_{00}^{5}$, which has cardinality $2^{28}$. The number of codewords which corresponds to the set $\cA_0^5$  is $2^{20}+2^{14}$. The number of codewords that correspond to $\cA_1^5\cup\cA_2^5\cup\cA_3^5$ is $\Gauss{7}{2}$.
\end{ex}


\vspace{-.3cm}
For small alphabets, when $q^2+q+1<\ell$, we use as the initial step for the recursion the Modified Pending Dots construction (Construction II in~\cite{et12}), where the last $n-3$ coordinates of the identifying vectors are partitioned into sets of size $q^2+q+2$ and then the same  idea for the construction of  the identifying vectors is applied in each such set. This Modified Pending Dots construction generates an $(n,M,2,3)_q$ constant dimension code
with
$M=q^{2(n-3)}+\sum_{i=1}^{\alpha}\Gauss{q^2+q+2}{2}q^{2(n-3-(q^2+q+2)i)}$,
which contains the lifted MRD code, where $\alpha=\left\lfloor\frac{n-3}{q^2+q+2}\right\rfloor$.
Then the size of an $(n,M,k-1,k)_q$ constant dimension code $\C^k$ obtained from the modified recursive construction is given by
$$|\C^k|= q^{2(n-k)}+\sum_{j=3}^{k-1} q^{2(n-\sum_{i=j}^k i)}+\sum_{i=1}^{\alpha_k}\Gauss{q^2+q+2}{2}q^{2(n-\frac{k^2+k-6}{2}-(q^2+q+1)i)},
$$
where $\alpha_k=\left\lfloor\frac{n-\frac{k^2+k-6}{2}}{q^2+q+2}\right\rfloor$.
Then, we obtain the following corollary.

\begin{cor}
Let $n\geq \frac{k^2+3k-2}{2}$ and $q^2+q+1< \ell$, where $\ell= n-\frac{k^2+k-6}{2}$ for odd $n-\frac{k^2+k-6}{2}$ (or $\ell= n-\frac{k^2+k-4}{2}$ for even $n-\frac{k^2+k-6}{2}$). Then
$$A_{q}(n, k-1, k) \geq q^{2(n-k)}+\sum_{j=3}^{k-1} q^{2(n-\sum_{i=j}^k i)}+\sum_{i=1}^{\alpha_k}\Gauss{q^2+q+2}{2}q^{2(n-\frac{k^2+k-6}{2}-(q^2+q+1)i)},$$
where $\alpha_k=\left\lfloor\frac{n-\frac{k^2+k-6}{2}}{q^2+q+2}\right\rfloor$.
\end{cor}

In the following, we compare the size of the codes obtained from Construction A (and its modification for small alphabets) to the bound in Theorem~\ref{trm:upper bound from Steiner Structure}.
In particular, we are interested in an estimation of the function $F(n,k,q)$ defined by $F(n,k,q):=\frac{\C^k-q^{2(n-k)}}{A_q(n-k,k-2,k-1)}$. The following bound on $A_q(n,d,k)$ was established in~\cite{wa03a,xi09,et11}:

\begin{equation*}
{A}_{q}(n,d,k)\leq\frac{\Gauss{n}{k-d+1}}{\Gauss{k}{k-d+1}}.
\label{eq:Johnson}
\end{equation*}
%
Then
\begin{eqnarray}
\label{eq:ratio}
&F(n,k,q)&=\frac{\C^k-q^{2(n-k)}}{A_q(n-k,k-2,k-1)}\geq \frac{\C^k-q^{2(n-k)}}{\begin{footnotesize}\Gauss{n-k}{2}/\Gauss{k-1}{2}\end{footnotesize}}.
\end{eqnarray}

One can show that $F(n,k,q)$ is an increasing function in $k$ and $q$ and that for $k\geq 10$, $n\geq \frac{k^2+3k-2}{2}$, it holds that $F(n,k,2)\geq 0.99$. Hence, Construction A asymptotically attains the bound of Theorem~\ref{trm:upper bound from Steiner Structure} for any $k$ and $q$. In fact it gets very close to the bound already for small values of $k$ and $q$.   In comparison, the lifted MRD construction attains the bound asymptotically as well, but is much further away from the bound for small parameters.

The comparison between the cardinality of codes obtained by Construction A and other known codes is given in Section~\ref{sec:Tables}, Table~\ref{table1}.

\section{Constructions for  $(n,M,2,k)_q$ Codes}
\label{sec:constructionsk=4}

In this section we present two constructions for $(n,M,2,k)_q$ codes with $k\geq 4$ and $n\geq 2k+2$.
These constructions will then give rise to new lower bounds on the size of constant dimension codes with minimum injection distance $2$ (or equivalently subspace distance $4$).
The first one (Construction B), which is a modification of the multilevel construction from~\cite{et09}, is based on a specific choice of a set of identifying vectors obtained from matchings and the complement of matchings of the corresponding complete graphs and is given for general $k\geq 4$.
The second one (Construction C) combines the results on pending blocks and
Ferrers diagrams arising from different (nearly) perfect matchings of the complete graph.
Since it improves the first construction only for the parameters $k=4$ and $k=5$, it will only be explained for these two cases.

\subsection{A Special Instance of the Multilevel Construction}

The multilevel construction (see Section \ref{sec:preliminaries}) is a general code construction which usually provides large codes. However, it does not give rise to a general formula for the cardinality of the arising codes, since this construction depends on the specific choice of a related constant weight code.
In the following, we will use a specific (nearly) perfect matching  and the complement of a matching  of  complete graphs of sizes $n-k$ and $k$,  respectively,
to produce a good choice of the constant weight code for the multilevel construction and
to get a closed formula for the constant dimension code cardinality. 

\vspace{.5cm}

We first need the following result, which is similar to Lemma~\ref{lem14}.


\begin{lem}
\label{trm:k-2 ones}
Let $n\geq 2k+2$.
 Let $v$ be an identifying vector of length $n$ and weight $k$, such that there are $k-2$ many ones in the first $k$ positions of $v$. Then the Ferrers diagram  arising from $v$ has more or equally many dots in the first row than in the last column,
 and the upper bound for the dimension of a Ferrers diagram code with minimum distance $2$ is the number of dots that are not in the first row.
\end{lem}

\begin{proof}
Analogous to the proof of Lemma~\ref{lem14}.
\end{proof}

From Theorems \ref{thm1} and \ref{thm2} the next statement follows.

\begin{cor}\label{cor:remove_row}
 The dimension of a Ferrers diagram code with minimum distance $2$ in the setting of Lemma~\ref{trm:k-2 ones} is the number of dots that are not in the first row.
\end{cor}

Let $n\geq 2k+2$ and define the subset of $\F_2^{n-k}$
\[O_{n-k} :=\{(110\dots 0), (00110\dots 0), (0000110\dots 0),\dots\} ,\]
which has $\lfloor  \frac{n-k}{2}\rfloor$ elements (the two ones are always shifted to the right by two positions).
In other words, if we denote by $v_i(j)$ the $j$th coordinate of the vector $v_i$, the set $O_{n-k}$ contains binary vectors $v_i$ of length $n-k$ and weight $2$, such that $v_i(j)=1$ if and only if $\lceil\frac{j}{2}\rceil=i$. Note, that for odd $n-k$ the last entry of all vectors in $O_{n-k}$ is always zero.

Also, we define the subset of  $\F_2^{k}$
\[\bar O_{k} :=\{(11\dots 100), (11\dots 10011), (11\dots 1001111),\dots\},\]
which has $\lfloor  \frac{k}{2}\rfloor$ elements (the two zeros are always shifted to the left by two positions).
In other words, the set $\bar O_{k}$ contains binary vectors $u_i$ of length $k$ and weight $k-2$, such that $u_i(j)=0$ if and only if $\lceil\frac{k-j+1}{2}\rceil=i$. Note, that for odd $k$ the first entry of all vectors in $\bar O_{k}$ is always one.

\begin{rem}
The elements of $O_{n-k}$ and $\bar O_{k}$ form a (nearly) perfect matching of $K_{n-k}$ and the complement of a (nearly) perfect matching of $K_{k}$, respectively.
\end{rem}


\noindent{\textbf{Construction B.}}

Let $n\geq 2k+2$.
We use the following sets of identifying vectors for the multilevel construction:

\begin{align*}
\cA_0^k &= \{(11\dots 11111||0 \dots 0) \}\\
\cA_1^k &= \{(11\dots 11100|| v) \mid v\in O_{n-k}\}\\
\cA_2^k &= \{(11\dots 10011|| v) \mid v\in  O_{n-k}\}\\
\vdots 
\end{align*}
\[
\cA_{\lfloor  \frac{k}{2}\rfloor}^k = \Big\{(w|| v) \mid v\in  O_{n-k} , w= \left\{\begin{array}{ll} (0011\dots 1) & \textnormal{ if } k \textnormal{ even } \\ (10011\dots 1) & \textnormal{ if } k \textnormal{ odd }\end{array}\right. \Big\},
\]
where the prefixes are the different elements from $\bar O_{k}$ (except for $\cA_0^k$).
 Then we construct the corresponding lifted FDMRD codes with injection distance $2$. Note that the code corresponding to $\cA_0^k$ is the conventional lifted MRD code. Furthermore we add the largest known $(n-k,M,2,k)_q$ code, with $k$ zero columns appended in front of every codeword.

\begin{thm}
\label{thm:cardinalityB}
The code from Construction B has minimum injection distance $2$ and cardinality
\[q^{(k-1)(n-k)} +A^*_q(n-k,2,k)+\]
\[\left(\sum_{i=0}^{\lfloor \frac{k-3}{2}\rfloor} q^{(k-3)(n-k)-4i} + \epsilon(k-1) q^{(k-3)(n-k-2)}\right) \sum_{i=0}^{\lfloor\frac{n-k}{2}\rfloor -1} q^{2(2i+\epsilon(n-k))} , \]
where $\epsilon(i)=1$ if $i$ odd and $\epsilon(i)=0$ if $i$ even.
\end{thm}
\begin{proof}
The minimum distance for elements with different identifying vectors follows by Proposition~\ref{prop3} from the Hamming distance of the identifying vectors, which is always at least $4$. For elements with the same identifying vector it follows from the minimum rank distance of the FDMRD code, by Proposition~\ref{pr:equal id}.

The cardinality can be shown as follows.
 From Theorem~\ref{thm2}, Lemma~\ref{trm:k-2 ones} and Corollary~\ref{cor:remove_row}
 we know that the number of dots not in the first row of the FD is the dimension of the FDMRD code. Hence, the subcode arising from $\cA_0^k$ has dimension $(k-1)(n-k)$. The number of matrix fillings for the height-$2$ Ferrers diagrams corresponding to $O_{n-k}$ is equal to $\sum_{i=0}^{\lfloor \frac{n-k}{2}\rfloor -1} q^{2(2i + \epsilon(n-k))}$ (where the empty matrix is also counted).
The number of fillings for the Ferrers diagrams corresponding to $\bar O_{k}$ without the first rows is equal to $\sum_{i=0}^{\lfloor \frac{k-3}{2}\rfloor} q^{(k-3)(n-k)-4i} + \epsilon(k-1) q^{(k-3)(n-k-2)}$.
Hence the formula follows.
\end{proof}
\vspace{-0.2cm}
\begin{cor}
Let $n\geq 2k+2$. Then
\[A_q(n,2,k) \geq
\sum_{i=1}^{\lfloor\frac{n-2}{k}\rfloor -1}\left( q^{(k-1)(n-ik)}+ \frac{(q^{2(k-2)}-1)(q^{2(n-ik-1)}-1)}{(q^4-1)^2}q^{(k-3)(n-ik-2)+4}\right).
\]
\end{cor}
\begin{proof}
From Theorem~\ref{thm:cardinalityB} it follows that the value for $A_q(n,2,k)-A^*_q(n-k,2,k)- q^{(k-1)(n-k)}$ is lower bounded by
\[\left(\sum_{i=0}^{\lfloor\frac{k-3}{2}\rfloor} q^{(k-3)(n-k)-4i}\right) \left(\sum_{i=0}^{\lfloor\frac{n-k-2}{2}\rfloor} q^{4i + 2\epsilon(n-k)}\right) = q^{(k-3)(n-k)+2\epsilon(n-k)} \left(\sum_{i=0}^{\lfloor\frac{k-3}{2}\rfloor} q^{-4i}\right) \left(\sum_{i=0}^{\lfloor\frac{n-k-2}{2}\rfloor} q^{4i }\right) .\]
Solving the sums and then using the equality $4\lfloor \frac{x}{2} \rfloor = 2x-2\epsilon(x)$ we get that this expression is equal to
\[q^{(k-3)(n-k)+2\epsilon(n-k)}  \frac{q^{-4 \lfloor\frac{k-3}{2} \rfloor} (q^{4(\lfloor\frac{k-3}{2}\rfloor+1)} -1)    (q^{4\lfloor\frac{n-k}{2}\rfloor} -1)   }{ (q^4-1)^2} \]
\[= q^{(k-3)(n-k)+2(\epsilon(n-k) + \epsilon(k-1))} \frac{ (q^{2(k-1-\epsilon(k-1))} -1)    (q^{2(n-k-\epsilon(n-k))} -1)   }{ (q^4-1)^2} \]
This expression takes its minimum for $\epsilon(n-k)=\epsilon(k-1)=1$, hence
$$A_q(n,2,k)\geq A^*_q(n-k,2,k)+ q^{(k-1)(n-k)} +q^{(k-3)(n-k-2)+4}  \frac{(q^{2(k-2)}-1)(q^{2(n-k-1)}-1)}{(q^4-1)^2}.
$$
Applying this bound recursively yields the desired formula.
\end{proof}
\vspace{-.2cm}
\begin{rem} Here we derived a closed cardinality formula for the special instance of the multilevel construction for $d=2$. Note that one can also apply this idea to obtain a bound on the cardinality for constant dimension codes with other values for the minimum injection distance.
\end{rem}

\subsection{New $(n,M,2,4)_q$- and $(n,M,2,5)_q$ Codes from One-Factorizations and Pending Dots}
The construction  presented in this subsection is based on a one-factorization of a complete graph which is used to construct a set of identifying vectors for the proposed codes, by generalizing the Pending Dots construction to  $
k>3$.
However, in contrast to  the Pending Dots construction, here we use not all but specifically chosen perfect matchings which result in a large constant dimension code.  First, we consider one-factorizations and the Ferrers diagrams arising from them.
\subsubsection{Ferrers Diagrams from One-Factorizations of the Complete Graph}
\label{sub:OF}
We will now present some results on Ferrers diagrams arising from the weight-$2$ vector representation of matchings of the complete graph $K_n$.
To do so we will use some graph theoretic results (see e.g.~\cite{tu84,li01b}) that will be useful for our choice of identifying vectors later on.
We start by with the existence proof of (near) one-factorizations, (see also Lemma~\ref{lm:1-factorization} in Section~\ref{sec:preliminaries}),
since we need the idea of this proof for our following results.

\begin{thm}[\cite{tu84,li01b}]\label{circleconstruction}
\mbox{}
\begin{enumerate}
\item If $n$ is odd there always exists a near one-factorization (NOF) of $K_n$.
\item If $n$ is even there always exists a one-factorization (OF) of $K_n$.
\end{enumerate}
\end{thm}
\begin{proof}
\mbox{}
\begin{enumerate}
\item
If $n$ is odd we can draw the nodes of $K_n$ as a circle. Then we can choose one edge and all its parallels, which will give us a nearly perfect matching 
of $K_n$. We can repeat this step for any edge that is not covered yet
and get a NOF of $K_n$.
\item
If $n$ is even we can use $n-1$ nodes of $K_n$ as a circle, just like before, and use the remaining node as the center of the circle. Then we use again the set of parallel edges plus the edge that connects the remaining node on the circle with the center of the circle, which is a perfect matching.
The set of all these different perfect matchings is an OF of $K_n$.
\end{enumerate}
\end{proof}

Then one can easily count the number of elements in the sets of a NOF or an OF of $K_n$ (see also Lemma~\ref{lm:1-factorization}):

\begin{lem}
\mbox{}
\begin{enumerate}
\item
For a given odd $n$ the NOF of $K_n$ has $n$ many nearly perfect matchings and each one of them contains $\frac{n-1}{2}$ elements.
\item
For a given even $n$ the OF of $K_n$ has $n-1$ many perfect matchings and each one of them  contains $\frac{n}{2}$ elements.
\end{enumerate}
\end{lem}


As in Section~\ref{sec:preliminaries},
we denote the different (nearly) perfect matchings of a (near) one-factorization in the vector representation by $P_i$ and call them \emph{classes}.

In the following construction we want to use the matchings which contribute the largest possible FDRM codes. So we need the following lemma, which gives the sizes of the corresponding Ferrers diagrams and, as a consequence, the cardinality of the FDRM codes.
We use  the construction of matchings described in the proof of Theorem~\ref{circleconstruction}. We denote $n':=n-k$ and label all the outside nodes counter-clock-wise from $1$ to $n'-1$ if $n'$ is even, and from $1$ to $n'$ if $n'$ is odd. If $n'$ is even, the center node is labeled by $n'$ and we name $P_i$ the perfect matching that contains the edge $(n',i)$ as the center edge (i.e. all other edges are orthogonal to this one). If $n'$ is odd, there is no center node and we name $P_i$ the nearly perfect matching that corresponds to the matching that does not cover node $i$.

\begin{lem}\label{thm6}
 For a given $P_i$, the size of the respective FDRM code with rank distance $1$ (i.e.\ the number of different matrix fillings for the corresponding Ferrers diagrams) is given by
\begin{itemize}
 \item
 $\left(\frac{n'}{2} - i\right) q^{(n'-2i)} + (i-1)q^{(2(n'-i)-1)} + q^{(n'-i-1)}$\\  if $i\leq \frac{n'}{2}$ and $n'$ is even,
 \item
$\left(i - \frac{n'}{2} \right) q^{(3n'-2(i+1))} + (n'-i-1)q^{(2(n'-i)-1)} + q^{(n'-i-1)} $  if $ i>  \frac{n'}{2}$
and $n'$ is even,
\item
$\left(\frac{n'+1}{2} - i\right) q^{(n'-2i-1)} + (i-1)q^{(2(n'-i)-1)} $\\ if $ i\leq \frac{n'+1}{2}$ and $n'$ is odd,
\item
$\left(i - \frac{n'+1}{2} \right) q^{(3n'-2i-1)} + (n'-i)q^{(2(n'-i)-1)}$\\  if $ i>  \frac{n'+1}{2}$
and $n'$ is odd.
\end{itemize}
\end{lem}

\begin{proof}
Can be found in Appendix~A.
\end{proof}

\subsubsection{Code Construction}

%
%
%
%
%
%
%
%
%
%
%


We will now describe a construction for constant dimension codes with $k=4$ and $k=5$.
The idea in both cases is similar to the multilevel construction: To construct the identifying vectors, we start with $(1\dots 1 0\dots 0)$ and then construct sets of identifying vectors with prefixes of length $k$ and weight $k-2$, and suffixes of length $n':=n-k$ and weight $2$. The suffixes will be chosen from some of the (nearly) one-factors $P_i$ of $K_{n-k}$. We choose the prefixes and suffixes that contribute the largest FDRM codes, using Lemma~\ref{thm6}. In addition, we use pending dots to allow for a choice of identifying vectors with a smaller Hamming distance.



\vspace{.3cm}

\noindent{\textbf{Construction C-4.}}
\\
Let $n\geq 10$ and $n'= n-4$. Hence, $n'$ is even if and only if $n$ is even.
We use the following sets of identifying vectors
\begin{align*}
\cA_0^4 &= \{(1111||0 \dots 0) \}\\
\cA_1^4 &= \{(1100|| v), (0011|| v) \mid v\in P_{\lceil\frac{n'}{2}\rceil+1}\}\\
\cA_2^4 &= \{(1001|| v), (0110|| v) \mid v\in P_2\}\\
\cA_3^4 &= \Big\{(1010|| v), (0101|| v) \mid  v\in \bigcup_{i=2}^{\min\{\lceil\frac{q}{2}\rceil+1,\lfloor\frac{n'}{2}\rfloor\}} P_{\lceil\frac{n'}{2}\rceil+i} \cup \bigcup_{i=3}^{\min\{\lfloor\frac{q}{2}\rfloor+2,\lceil\frac{n'}{2}\rceil\}} P_i  \Big\}
\end{align*}
and construct the corresponding lifted FDMRD codes  with injection distance $2$, where we use the pending dot in $\cA_3^4$. Note that the code corresponding to $\cA_0^4$ is the conventional lifted MRD code. Furthermore, we add the largest known $(n-4,M,2,4)_q$ code, with $4$ zero columns appended in front of every codeword, to obtain a constant dimension code $\C^4$.

\begin{thm}\label{thm:c4}
The code $\C^4$ obtained by Construction C-4 has minimum subspace distance $4$ and cardinality given by
\[ q^{3(n-4)} + (q^{(n-4)}+ q^{(n-6)}) \left[q^{2(n-6)}+ (\frac{n}{2}-4)q^{(n-7)} +q^{(\frac{n}{2}-4)} \right] + (q^{(n-5)}+ q^{(n-6)})\times \]
\[\left[    \sum_{i=2}^{\min\{\lceil\frac{q}{2}\rceil+1,\lfloor\frac{n'}{2}\rfloor\}} (iq^{2n-2i-10} + (\frac{n-6}{2} -i) q^{n-2i-5}   + q^{\frac{n-6}{2} - i}) + \right.  \]
\[ \left.
\sum_{i=1}^{\min\{\lfloor\frac{q}{2}\rfloor+1,\lceil\frac{n'}{2}\rceil -1\}} (iq^{2n-2i-11} + (\frac{n-6}{2} -i) q^{n-2i-6}   + q^{n - i-6} )  \right] \]
\[   +A_q^*(n-4,2,4)    \]
if $n$ is even, and
\[q^{3(n-4)} +   (q^{(n-4)}+ q^{(n-6)}) \left[q^{2(n-6)}+ (\frac{n-3}{2})q^{(n-8)} \right] + (q^{(n-5)}+ q^{(n-6)}) \times \]
\[\left[\sum_{i=2}^{\min\{\lceil\frac{q}{2}\rceil+1,\lfloor\frac{n'}{2}\rfloor\}} (iq^{2n-2i-10} + (\frac{n-5}{2} -i) q^{n-2i-6} )   \right.   \]
\[ +   \left.
\sum_{i=1}^{\min\{\lfloor\frac{q}{2}\rfloor+1,\lceil\frac{n'}{2}\rceil -1\}} (iq^{2n-2i-11} + (\frac{n-5}{2} -i) q^{n-2i-7} )    \right]\]
\[  +A_q^*(n-4,2,4)    \]
if $n$ is odd.

\end{thm}
\begin{proof}
The minimum distance for elements with different identifying vectors follows from the Hamming distance of the identifying vectors, together with the pending dots, i.e, from Proposition~\ref{prop3} and Lemma~\ref{lm:pending dots}. For elements with the same identifying vector it follows from the minimum rank distance of the FDMRD code, by Proposition~\ref{pr:equal id}.

The proof for the cardinality can be found in Appendix~B.
\end{proof}
\begin{ex}
Let $q=2$, $n=10$. Then we have $A_2(10,2,4)\geq 2^{18}+37456+21$, where $A_2(6,2,4)=21$. The largest previously known code obtained by the multilevel construction~\cite{et09} has cardinality $2^{18}+34768$.
\end{ex}
\vspace{-0.4cm}
\begin{ex}
Let $q=2$, $n=12$. Then we have $A_2(12,2,4)\geq 2^{24}+2333568+701+2^{12}=2^{24}+2338365$, where $A_2(8,2,4)\geq 701+2^{12}$. The largest previously known code obtained by the multilevel construction~\cite{et09} has cardinality $2^{24}+2290845$.
\end{ex}

\noindent{\textbf{Construction C-5.}}
\\
Let $n\geq 12$ and $n'= n-5$. Hence, $n'$ is even if and only if $n$ is odd.
We use the following sets of identifying vectors
\begin{align*}
\cA_1^5 &= \{(11100|| v), (10011|| v) \mid v\in P_{\lceil\frac{n'}{2}\rceil+1}\}\\
\cA_2^5 &= \{(11010|| v), (01101|| v) \mid v\in P_2 \}\\
\cA_3^5 &= \{(01110|| v), (10101|| v) \mid v\in P_{\lceil\frac{n'}{2}\rceil+2} \}\\
\cA_4^5 &= \{(00111|| v), (11001|| v) \mid v\in P_3\}\\
\cA_5^5 &= \Big\{(10110|| v), (01011|| v) \mid v\in \bigcup_{i=3}^{\min\{\lceil\frac{q}{2}\rceil+2,\lfloor\frac{n'}{2}\rfloor\}} P_{\lceil\frac{n'}{2}\rceil + i}  \cup \bigcup_{i=4}^{\min\{\lfloor\frac{q}{2}\rfloor+3,\lceil\frac{n'}{2}\rceil\}} P_i  \Big\}
\end{align*}
and construct the corresponding lifted FDMRD codes with injection distance 2, where we use the pending dot in $\cA_5^5$. Note that the code corresponding to $\cA_0^5$ is the conventional lifted MRD code. Furthermore, we add the largest known $(n-5,M,2,5)_q$ code, with $5$ zero columns appended in front of every codeword to obtain a constant dimension code $\C^5$.

\begin{thm}\label{thm:c5}
The code $\C^5$ obtained by Construction C-5 has minimum subspace distance $4$ and cardinality given by
\begin{itemize}
\item
$q^{4(n-5)}  + (q^{2n-10}+ q^{2n-14}) (q^{2(n-7)}+ (\frac{n-8}{2})q^{(n-9)})  +
(q^{2n-11}+ q^{2n-13}) (\frac{n-8}{2}  q^{(n-10)} + q^{(2n-15)} ) +
(q^{2n-12}+ q^{2n-13}) (2q^{2(n-8)}+ \frac{n-10}{2}q^{(n-11)} )+ (q^{2n-12}+ q^{2n-14}) 
\left[\sum_{i=3}^{\min\{\lceil\frac{q}{2}\rceil+2,\lfloor\frac{n'}{2}\rfloor\}} (iq^{2n-2i-12} + (\frac{n-6}{2} -i) q^{n-2i-7} )       + \right. \\ \left.     \sum_{i=2}^{\min\{\lfloor\frac{q}{2}\rfloor+2,\lceil\frac{n'}{2}\rceil -1\}} (iq^{2n-2i-13} + (\frac{n-6}{2} -i) q^{n-2i-8} )  \right]
+  A^{*}_q(n-5,2,5)   $\\
if $n$ is even.\\
\item
$q^{4(n-5)}  + (q^{2n-10}+ q^{2n-14}) (q^{2n-14}+ (\frac{n-9}{2})q^{(n-8)} + q^{\frac{n-9}{2}})  +
(q^{2n-11}+ q^{2n-13}) (\frac{n-9}{2}  q^{(n-9)} + q^{(2n-15)} + q^{n-8} ) +
(q^{2n-12}+ q^{2n-13}) (q^{2n-16}+ (\frac{n-11}{2})q^{(n-10)} + q^{\frac{n-11}{2}}  )  + (q^{2n-12}+ q^{2n-14})  \left[\sum_{i=3}^{\min\{\lceil\frac{q}{2}\rceil+2,\lfloor\frac{n'}{2}\rfloor\}} (iq^{2n-2i-12} +\right. \\  (\frac{n-7}{2} -i) q^{n-2i-6} +
 q^{\frac{n-7}{2} -i})       +  \sum_{i=2}^{\min\{\lfloor\frac{q}{2}\rfloor+2,\lceil\frac{n'}{2}\rceil -1\}}(iq^{2n-2i-13} +  (\frac{n-7}{2} -i) q^{n-2i-7} + q^{n -i-7})  \Big]
+  A^*_q(n-5,2,5 )   $\\
if $n$ is odd.
\end{itemize}
\end{thm}
\begin{proof}
The minimum distance for elements with different identifying vectors follows from the Hamming distance of the identifying vectors, together with the pending dots, by Proposition~\ref{prop3} and Lemma~\ref{lm:pending dots}. For elements with the same identifying vector it follows from the minimum rank distance of the FDMRD code, by Proposition~\ref{pr:equal id}.

The proof for the cardinality can be found in Appendix~B.
\end{proof}

\begin{rem}
One can easily generalize Constructions C-4 and C-5 to larger values of $k$ by choosing the prefixes for the sets $\cA_i^k$ as follows: Choose an OF (or NOF) of $K_{k}$, look at its vector representation and add the all-one vector to all these vectors (i.e.\ bitflip all coordinates). Thus, the prefixes in a given set $\cA_i^k$ form a code with constant weight $k-2$ and minimum Hamming distance $4$ in $\F_2^k$. But one can then prove that there is no such set with pending dots in all its elements. Hence, this generalization would not improve the multilevel construction from \cite{et09}. This is why we only describe the construction for $k=4$ or $k=5$ in this work.
\end{rem}

The comparison between the multilevel construction and the codes obtained by Constructions B and C can be found in Section~\ref{sec:Tables}, Table~\ref{table2}.
One can see that Constructions C-4 and C-5 improve Construction B, but remember that Construction B works for general $k$ and thus for more parameters than Construction C-4 or C-5. Furthermore, Construction C-4 yields larger codes than the multilevel construction and hence results the largest known codes for some parameter sets. On the other hand, Construction C-5 does not improve the cardinality of the codes arising from the multilevel construction. The advantage still is that we have a closed formula for all constructions explained in this section, in contrast to the multilevel construction.



\section{Construction for a New $(n,M,d,k)_q$ Code from an Old Code}
\label{sec:new-old}

In the following we discuss a way for constructing a  new constant dimension code with minimum injection distance $d$ (or subspace distance $2d$) from a given one. This approach is fairly simple, but surprisingly, for some families of parameters it provides the largest known codes.

\vspace{0.2cm}

\noindent\textbf{Construction D.}
\\
Let $\C\in \Gr$ be an $(n,M,d,k)_q$ code, let $\Delta$ be an integer such that $\Delta\geq k$, and let $\cC$ be an $[\cF,\Delta(k-d+1), d]$ FDMRD code with a full $k\times \Delta$ rectangular Ferrers diagram.
Define
$$\C'=\{X'\in\cG_q(k,n'):\mbox{RE}(X')=[\mbox{RE}(X)A], X\in \C, A\in \cC\}.$$

\begin{thm}
\label{trm:code extension}
The code $\C'$ obtained by Construction D is an $(n'=n+\Delta , M', d,k)_{q}$ code in $\cG_q(k,n')$, such that
$$M'= Mq^{\Delta (k-d+1)}.$$
\end{thm}

\begin{proof}
Since $|\cC|=q^{\Delta(k-d+1)}$, it follows from Theorem~\ref{thm1} that $M'= Mq^{\Delta (k-d+1)}$.
To prove the minimum distance we distinguish between two cases:
\begin{enumerate}
\item
Let $X',Y'\in\C'$, such that $\mbox{RE}(X')=[\mbox{RE}(X)A]$, $\mbox{RE}(Y')=[\mbox{RE}(X)B]$, for $X\in\C$ and $A,B\in \cC$, $A\neq B$. Then $v(X')=v(Y')$ since all the ones of the identifying vectors of the codewords from $\C'$ appear in the first $k$ coordinates. Hence, by Proposition~\ref{pr:equal id}, $d_I(X',Y')=d_R(\mbox{RE}(X'),\mbox{RE}(Y'))$. Since $\mbox{RE}(X')-\mbox{RE}(Y')=[0 A-B]$, where $0$ is a $k\times n$ zeros matrix, we have $d_I(X',Y')=d_R(A,B)\geq d$, since $A,B\in\cC$.
\item
Let $X',Y'\in\C'$, such that $\mbox{RE}(X')=[\mbox{RE}(X)A]$, $\mbox{RE}(Y')=[\mbox{RE}(Y)B]$, for $X,Y\in\C$, $X\neq Y$, and $A,B\in \cC$. Then $d_I(X',Y')=k-\dim(X'\cap Y')\geq k-\dim(X\cap Y)\geq d$, since $X,Y\in \C$.
\end{enumerate}
\end{proof}


\begin{ex}
We take the  $(8,2^{12}+701,2,4)_2$  code $\C$ constructed in~\cite{et12} and apply on it Construction~D
with $\Delta =4$.
Then the new code $\C'$ has cardinality $|\C'|= 2^{24}+701\cdot 2^{12}=2^{24} + 2871296$ and has parameters $(12,|\C'|, 2,4)_2$.
The largest previously known code of these parameters of size $2^{24}+2290845$ was obtained in~\cite{et09}.
\end{ex}

Like in the constructions before we can then also add codes of shorter length with zeros appended in front to these codes. Hence we get a new lower bound as follows.

\begin{cor}
Let $n\geq 3k$. Then for any positive integer  $\Delta$, such that
$n\geq \Delta  \geq k$, it holds that
\[A_q(n,d,k) \geq q^{\Delta (k-d+1)} A_q(n-\Delta ,d,k) + A_q(\Delta,d,k)  .\]
In particular, for $\Delta =k$, we get
\[A_q(n,d,k) \geq q^{k(k-d+1)} A_q(n-k,d,k) + 1\]
and, for $\Delta  = n-k$, we get
\[A_q(n,d,k) \geq q^{(n-k)(k-d+1)} +  A_q(n-k,d,k) \]
which, if recursively solved, corresponds exactly to the formula of the multi-component lifted MRD codes from \cite{tr13phd}.
\end{cor}

\begin{rem} Note that Construction D  is related to the interleaved rank-metric codes (see e.g.\ \cite{wa13p}). In particular, the code obtained in Construction D can be considered as a lifted Ferrers diagram interleaved code, where to the FDRM code raised from the first $n$ coordinates is appended another FDRM code with the same minimum rank distance and with a full rectangular $k\times \Delta$ Ferrers diagram. Then, this construction can be considered as a generalization of an interleaved construction, since every code can be used as the initial step of construction.
\end{rem}


\section{Comparison of Constant Dimension Code Sizes}
\label{sec:Tables}

In this section we compare the cardinalities of our new code constructions to other known constant dimension code constructions. Since Constructions A--C are defined for $d=k-1$ and $d=2$ we will only cover these two cases. The largest previously known general construction with a closed cardinality formula is the multicomponent construction~\cite{ga11,tr13phd}. This construction is a special case of the multilevel construction of~\cite{et09} (see Section \ref{sec:preliminaries}), where we require the $k$ ones in the identifying vectors to be in one block of length $k$. Then the arising Ferrers diagrams are full rectangles and can be filled with an MRD code. A closed cardinality formula for this construction was derived in~\cite[Theorem 2.9]{tr13phd} as
\begin{equation}
\label{eq:MC}
\sum_{i=0}^{\lfloor \frac{n-2k}{d}\rfloor} q^{(k-d+1)(n-k-d i)} +  \sum_{i=\lfloor \frac{n-2k}{d}\rfloor +1 }^{\lfloor \frac{n-k}{d}\rfloor} \lceil q^{k(n-k+1-d (i+1))} \rceil
\end{equation}

In the following lemma we compare our Construction A with the multicomponent construction and give a lower bound on the difference between the respective cardinalities. The proof of the lemma can be found in Appendix~C.

\begin{lem}
\label{lm:diff A-MC}
 Let $n\geq \frac{k^2+3k-2}{2}$.
 Let $C_A$ be an $(n,|C_A|, d=k-1, k)$ code obtained by Construction A and let $C_{MC}$ be an $(n,|C_{MC}|, d=k-1, k)$ code obtained by the multicomponent construction.
  Then
$$|C_A|-|C_{MC}|> q^{2n-k^2-k+1}.
$$
\end{lem}

In the following lemma we compare the cardinality of a code obtained by Construction B with the cardinality of a code obtained by the
multicomponent construction.  The proof of the lemma can be found in Appendix~D.
\begin{lem}
\label{lm:diff B-MC}
 Let $n\geq 2k+2$.
 Let $C_B$ be an $(n,|C_A|, d=2, k)$ code obtained by Construction B and let $C_{MC}$ be an $(n,|C_{MC}|, d=2, k)$ code obtained by the multicomponent construction.  Then
$$|C_B|-|C_{MC}|> q^{(k-1)(n-k)-8}.
$$
\end{lem}

In the following lemma we provide a comparison between Construction B and Construction C. For simplicity, we only consider the case $k=4, q=2$ and even $n$, but for $k=5$ and general $q, n$ the statement is similar.
The proof of the lemma can be found in Appendix E.
\begin{lem}
 \label{lm:diff B-C}
 Let $q=2$ and $n\geq 10$ be an even number. Let $C_B$ be an $(n,|C_B|,2,4)_2$ code obtained by Construction~B and let $C_C$ be an $(n,|C_C|,2,4)_2$ code obtained by Construction C-4. Then
$$|C_C|-|C_B|\geq 3\cdot 2^{3n-20}.$$
\end{lem}

%
Note that since the cardinality of a code obtained by Construction D depends on the choice of the base code, we cannot provide a closed cardinality formula and hence also no analytical comparison of Construction D with other known constructions.

If we do not require a closed cardinality formula, the largest known codes for most parameter sets arise from the multilevel construction
 with a lexicode as the set of identifying vectors \cite{et09}.
Tables \ref{table1} -- \ref{table2} show some examples of code cardinalities of the different constructions from this paper compared to the multilevel construction  and the multicomponent construction. The bold value for each line shows the largest cardinality for the given parameters.

For Construction A we use the cardinality formula of Theorem~\ref{trm:recursive_parameters}, for Construction B the formula of Theorem~\ref{thm:cardinalityB}.
For the values of Construction C we use the formulas of Theorems~\ref{thm:c4} and \ref{thm:c5},  for $k=4$ and $k=5$ respectively.
For Construction D  we use the respective multilevel codes (see \cite{et09}) of length $2k$  (i.e., $\Delta=n-2k$), and the $(8,4797,2,4)$ code from~\cite{et12}, as the old code from which we construct a new code. The cardinality formula for Construction D can be found in Theorem \ref{trm:code extension}.

All the $(n,M,d,k)_q$ codes presented in these tables contain a lifted MRD code of size $q^{(n-k)(k-d+1)}$, so the cardinalities of the constructed codes are written in the form $q^{(n-k)(k-d+1)}+(M-q^{(n-k)(k-d+1)})$.



\begin{table*}[ht]
{\centering
\begin{tabular}{|c|c|c||c|c|}
\hline
$(n,d,k)_q$ & A & D & multilevel & multicomponent  \\\hline
$(13,3,4)_2$ & $\mathbf{2^{18}+4747}$ &$2^{18}+4096$ & $2^{18}+4357$ & $2^{18}+4113$\\
$(14,3,4)_2$ & $\mathbf{2^{20}+19051}$ &$2^{20}+16384$ & $2^{20}+17204$ & $2^{20}+16641$\\
$(15,3,4)_2$ &$\mathbf{2^{22}+76331}$ &$2^{22}+65536$ & $2^{22}+68378$ & $2^{22}+66561$\\
\hline
$(19,4,5)_2$ & $\mathbf{2^{28}+1067627}$  & $2^{28}+1048576$ & $2^{28}+1052778$ & $2^{28}+1052673$\\
$(20,4,5)_2$ & $\mathbf{2^{30}+4270635}$ & $2^{30}+4194304$ & $2^{30}+4211044$ & $2^{30}+4210689$\\
\hline
$(19,4,5)_3$ & $\mathbf{3^{28}+3491666833}$ &$3^{28}+3486784401$ & $3^{28}+3487316403$ & $3^{28}+3487315843$\\
$(20,4,5)_3$ & $\mathbf{3^{30}+31425002590}$ &$3^{30}+31381059639$ & $3^{30}+31385846853$ & $3^{30}+31385842579$\\
\hline
\end{tabular}
\caption{Comparison of cardinalities of codes constructed according to Constructions A and D with the multilevel and the multicomponent construction.}\label{table1}
}
\end{table*}


\begin{table*}[ht]
{\centering
\begin{tabular}{|c|c|c|c||c|c|}
\hline
$(n,d,k)_q$ & B & C & D & multilevel & multicomponent  \\\hline
$(10,2,4)_2$ & $2^{18}+21861$ & $\mathbf{2^{18}+ 37477}$ &   --& $2^{18}+35685$ & $2^{18}+4113$\\
$(11,2,4)_2$ & $2^{21}+175024$ &$\mathbf{2^{21}+ 293200} $ & --&$2^{21}+285889$ &$2^{21}+33025$\\
$(12,2,4)_2$ & $2^{24} +1402877$ & $2^{24} + 2338365$ &  $\mathbf{2^{24}+2871296}$ & $2^{24}+2290845$ & $2^{24}+266257$ \\
$(13,2,4)_2$ &$2^{27}+11221585$ & $2^{27} +18517073$ & $\mathbf{2^{27}+22970368}$&$2^{27}+18328921$ &$2^{27}+2130177$ \\
\hline
\hline
$(12,2,5)_2$ & $2^{28} + 19009577$ & $2^{28} + 29377577$ &  -- & $\mathbf{2^{28}+30877839}$ & $2^{28}+1049601$\\
$(13,2,5)_2$ & $2^{32} + 304223372$ & $2^{32} + 447026316$ &  -- &$\mathbf{2^{32}+ 494999563}$ &$2^{32}+ 16810017$ \\
$(15,2,5)_2$ &$ 2^{40} +77883166687$ & $ 2^{40} + 113061122015 $ & $2^{40}+124519448576$& $\mathbf{2^{40}+126773908793}$&$2^{40}+ 4311777313 $\\
$(16,2,5)_2$ &$ 2^{44} + 1246130688803$ & $ 2^{44} + 1903760855843 $ & $2^{44}+1992311177216$ & $\mathbf{2^{44}+ 2028469279328}$& $2^{44}+   68988961793 $\\
\hline
%
\end{tabular}
\caption{Comparison of cardinalities of codes constructed according to Constructions B,  C-4, C-5, and D with the multilevel and the multicomponent construction.\label{table2}
}
}
\end{table*}

One can see that Construction A always results in the largest cardinality for a valid set of parameters (remember that Construction A is only defined for $d=k-1$). Furthermore, Construction C-4 beats the multilevel construction, whereas Construction C-5 does not for the parameter sets we used. Moreover, Construction D yields the largest known  $(12,2,4)_2$ and  $(13,2,4)_2$ codes. Note that Construction D is not defined for the parameters $(n,k)\in\{(10,4),(11,4), (12,5), (13,5)\}$, since $\Delta = n-2k < k$ in these cases.

Overall, our new constructions presented in this paper beat the known constructions for many sets of parameters. Note that, by construction, we cannot expect Construction B to improve on the cardinality of the multilevel construction. We still wanted to describe this construction to derive a closed cardinality formula, in contrast to the multilevel construction, for which no such formula exists.



\section{Non-Constant Dimension Codes}
\label{sec:non-constant}

%
%
%


In this section we consider codes in $\PG$ which are not constant dimension codes. Constructions of such codes were considered for the subspace metric in~\cite{et09,kh09p} and for the injection metric in~\cite{kh09p}. A code in the projective space can be considered as a union of constant dimension codes with different dimensions. Moreover, a construction of a code in $\PG$ can be done in a multilevel manner, i.e., first, the identifying vectors of the subspaces are chosen and then the corresponding lifted Ferrers diagrams rank-metric codes are constructed~\cite{et09,kh09p}. For this recall Proposition \ref{prop3}, which states that for any $X,Y \in \PG$,
\[d_S(X,Y) \geq d_H(v(X), v(Y)), \]
\[d_I(X,Y) \geq d_{asym}(v(X), v(Y)).\]

One can easily see that the largest constant dimension component of the final code  is of dimension $k=\lfloor \frac{n}{2}\rfloor$. Hence, to  construct a code in the projective space one can  start  by first choosing a constant dimension code with the minimum injection distance $d$ in $\mathcal{G}_q(\lfloor \frac{n}{2}\rfloor, n)$, then add codes with the same minimum distance in $\mathcal{G}_q(\lfloor \frac{n}{2}\rfloor \pm d, n)$, then in $\mathcal{G}_q(\lfloor \frac{n}{2}\rfloor \pm 2d, n)$, etc. The union is a projective space code in $\PG$ with minimum distance $d$. This is independent of the underlying metric, i.e., it works for both the subspace and the injection distance.


%
We will show that by using the codes (lower bounds) obtained in the previous sections, one can provide new  large codes in the projective space (and hence new lower bounds), for both the subspace and the injection metric.
To provide large codes in projective spaces we use the puncturing approach, presented in~\cite{et09}. Although the puncturing method was proposed for the subspace metric, we show that when applied on large constant dimension codes, it results in large codes  also for the injection metric. This  shows that puncturing is a powerful method to construct large codes for the injection metric as well.

First, we briefly describe the puncturing method presented in~\cite{et09}.
Let $X\in \Gr$  be a subspace which does not contain  the  $i$th unit vector vector $e_i\in \F_q^n$. The $i$-coordinate puncturing of $X$, denoted by $\Gamma_i(X)$, is the subspace in $\cG_{q}(k,n-1)$ obtained from $X$ by deleting the $i$th coordinate from each vector of $X$.
Let $1\leq\tau\leq n$ be the unique zero position of $v(Q)$, for a given $Q\in \cG_q(n-1,n)$ and let $v\in\F_q^n$ such that $v\notin Q$. Let $\C$ be an $(n,M,d)^S_q$ code in $\PG$ of subspace distance $d$, such that there exist codewords $X_1,X_2\in \C$ with $X_1\subseteq Q$ and $v\in X_2$. Then the \emph{punctured} code $\C'_{Q,v}$, defined by
\begin{eqnarray}
 \C'_{Q,v} {}= &\{\Gamma_{\tau}(X):X\in\C,X\subseteq Q\} \cup \nonumber\\
& \hspace{0.2cm} \; \{\Gamma_{\tau}(X\cap Q):X\in\C,v\in X\}, \nonumber
\end{eqnarray}
is a code in $\cP_q(n-1)$ with minimum subspace distance $d-1$, i.e., an $(n-1, M',d-1)_q^S$ code.
If $\C$ is a constant dimension code in $\Gr$, then the punctured code contains subspaces of dimensions $k$ and $k-1$.

The following lemma considers the minimum \emph{injection} distance of a punctured code of a constant dimension code.

\begin{lem}
\label{lem:puncturing injection}
Let $\C\in\Gr$ be a code with  minimum injection distance $d$, i.e., an $(n, M, d,k)_q$ constant dimension code. Let $Q\in \cG_q(n-1,n)$ and  $v\in\F_q^n$, $v\notin Q$, such that there exist two codewords $X_1,X_2\in \C$ with $X_1\subseteq Q$ and $v\in X_2$. Then the punctured code $\C'_{Q,v}$ has minimum injection distance~$d$, i.e., it is an $(n-1, M',d)_q^I$ code.
\end{lem}
\begin{proof} Since for any two subspaces $X,Y$ of the same dimension it holds that $d_I(X,Y)=d_S(X,Y)/2$, it is sufficient to check two subspaces $X,Y\in \C'_{Q,v}$ of different dimensions $k$ and $k-1$:
\[d_{I}(X,Y)=k-\dim(X\cap Y)=\frac{2k-2\dim(X\cap Y)}{2}\]
\[=\frac{d_S(X,Y)+1}{2}\geq \frac{(2d-1)+1}{2}=d.
\]
\end{proof}

The lower bound on the cardinality of the punctured code is given in the following theorem~\cite{et09}:

\begin{thm}
\label{thm:punctured_size}
 If $\C$ is an $(n, M, d,k)_q$ constant dimension code then there exists an $(n-1)$-dimensional subspace $Q$ and a vector $v\notin Q$, such that
\[|\C'_{Q,v}|\geq M\frac{q^{n-k}+q^k-2}{q^n-1}.
\]
\end{thm}


Now we present a construction for codes in the projective space. This construction generalizes the constructions for non-constant dimension codes from~\cite{et09,kh09p}.

\textbf{Construction of codes in projective space.}
 Let $\C\in \cG_q(\lfloor\frac{n+1}{2}\rfloor,n+1)$ be a constant dimension code of minimum injection distance $d_I=d$.  Let $\C'$ be the code obtained by puncturing $\C$. $\C'$ contains subspaces of $\F_q^{n}$ of dimensions $\lfloor\frac{n+1}{2}\rfloor$ and $\lfloor\frac{n+1}{2}\rfloor-1$ and has minimum subspace distance $2d-1$ and minimum injection distance $d$, by Lemma~\ref{lem:puncturing injection}.
\begin{itemize}
  \item For the injection metric, we add to $\C'$ the codewords of the largest known constant dimension codes with minimum injection distance $d$ from $\cG_q(\lfloor\frac{n+1}{2}\rfloor-1-id,n)$, for $i=1,\ldots, \lfloor\frac{\lfloor\frac{n+1}{2}\rfloor-1}{d}\rfloor$ and from $\cG_q(\lfloor\frac{n+1}{2}\rfloor+id,n)$, for $i=1,\ldots, \lfloor\frac{n-\lfloor\frac{n+1}{2}\rfloor}{d}\rfloor$.  The resulting code $\widetilde{\C_I}$ is a code in $\cP_q(n)$ with minimum injection distance $d$.
  \item  For the subspace metric, we add to $\C'$ the codewords of the largest known constant dimension codes with minimum subspace distance $2d$ from $\cG_q(\lfloor\frac{n+1}{2}\rfloor-1-i(2d-1),n)$, for $i=1,\ldots, \lfloor\frac{\lfloor\frac{n+1}{2}\rfloor-1}{2d-1}\rfloor$ and from $\cG_q(\lfloor\frac{n+1}{2}\rfloor+i(2d-1),n)$, for $i=1,\ldots, \lfloor\frac{n-\lfloor\frac{n+1}{2}\rfloor}{2d-1}\rfloor$. The resulting code $\widetilde{\C_S}$ is a code in $\cP_q(n)$ with minimum subspace distance $2d-1$.
\end{itemize}

\begin{rem} The cardinality of the code obtained by the above construction is lower bounded by using the results from the previous sections and by Theorem~\ref{thm:punctured_size}.
\end{rem}

We illustrate the idea of the construction for projective space codes based on the puncturing method, for both the subspace and the injection metric, in the following example.

\begin{ex} Let $q=2$ and $n=11$. First, let $\C\in\cG_2(6,12)$ be a constant dimension code with minimum injection distance $d_I=2$ and  size $1196288829$, obtained by the multilevel construction~\cite{et09}. By puncturing it, we can obtain a code in $\cP_2(11)$ of size at least $36808900$ (by Theorem~\ref{thm:punctured_size}), which includes subspaces of dimensions 5 and 6 of $\F_2^{11}$,  and has  minimum subspace distance $d_S=3$ and minimum injection distance $d_I=2$.
\begin{enumerate}
  \item We add the codewords of constant dimension codes with minimum injection distance $2$ from $\cG_2(1,11)$, $\cG_2(3,11)$, $\cG_2(8,11)$, $\cG_2(10,11)$ of sizes 1, 76331, 76331, 1, respectively. The final code $\widetilde{\C_I}$ has minimum injection distance  $d_I=2$  (and subspace distance $d_S=2$) and size $|\widetilde{\C_I}|=36961564$, such that $\log (|\widetilde{\C_I}|)=25.1395$ (compare to $24.63210$ in~\cite{kh09master}).
  \item We add the codewords of constant dimension codes with minimum injection distance $2$ from $\cG_2(2,11)$, $\cG_2(9,11)$ of size $681$ each. The final code $\widetilde{\C_S}$ has minimum subspace distance $d_S=3$ (and injection distance $d_I=2$) and size $|\widetilde{\C_S}|=36810200 $, such that $\log (|\widetilde{\C_S}|)=25.1336$.
\end{enumerate}

\end{ex}

%

%
%
%

Table~\ref{tabNonConst}  shows some examples of cardinalities of our codes based on puncturing (for both the subspace and the injection metric) in $\PG$ compared to the codes of~\cite{kh09master} (for the injection metric), for $q=2$. To make the comparison easier we present the cardinalities in the logarithmic form.

One can see that for odd $n$ our codes are larger than the known ones, while for even $n$ this is not the case.

\begin{table}[ht]
{\centering
\begin{tabular}{|c|c|c||c|c|c|}
\hline
$n$ & $d_S$ & $\log(\widetilde{\C_S})$ & $d_I$ & $\log(\widetilde{\C_I})$ & $\log(\widetilde{\C})$ from~\cite{kh09master} \\\hline
11 & 3 & 25.1336 &2& \textbf{ 25.1395}&24.6321 \\
11 & 5 & 18.9806 &3& \textbf{ 18.9806}& 18.0298\\
\hline
12 & 3 & 29.728 & 2 & 29.7586&\textbf{ 30.3372} \\
12 & 5 & 20.6101 &3& 20.6107 &\textbf{ 24.0054}\\
\hline
13 & 3 & 36.1454 &2& \textbf{ 36.1511}&35.6303\\
13 & 5 & 28.9917 &3& \textbf{ 28.9924}& 28.0265\\
\hline
14 & 3 & 41.7352 &2& 41.7651&\textbf{ 42.33625} \\
14 & 5 & 33.5804 &3& 33.5806&\textbf{ 35.00464} \\
\hline

\end{tabular}
\caption{Comparison of cardinalities of our codes in $\PG$ based on puncturing
with the codes in~\cite{kh09master}
.\label{tabNonConst}
}
}
\end{table}



\section{Conclusion and Open Problems}
\label{sec:conclusion}

In this work we presented new constructions for constant dimension codes, and based on these also new constructions for non-constant dimension codes. To do so we used the known techniques of the multilevel construction and pending dots, as well as new results on Ferrers diagrams arising from matchings of the complete graph. Moreover, we derived a way of constructing new codes from old codes. The new constructions give rise to the largest known codes for most sets of parameters, as shown in the tables of Section \ref{sec:Tables} and \ref{sec:non-constant}. This means that these codes have the best known transmission rate for a given error-correction capability.

For future research it would be interesting to derive bounds analogous to the one of Theorem \ref{trm:upper bound from Steiner Structure} for other values of $d$, and see if any of our constructions attain such a bound (asymptotically). Furthermore, we would like to develop results of Ferrers diagrams rank metric codes related to the complete graph for codes of minimum rank distance $d\neq 2$, and investigate if we could use such results for constant dimension code constructions with minimum injection distance $d$ (and respective non-constant dimension codes).

Another open question is how these codes can be decoded efficiently. Due to their similarity to the multilevel construction the codes constructed by our new constructions can be decoded with an analogous decoding algorithm but the structure of the identifying vectors might be useful and could be exploited for a more efficient algorithm.


\section*{Appendix}

\subsection{Proof of Lemma \ref{thm6}}
\label{ap:A}

\begin{proof}
We will prove the first statement for $n'$ even and $i\leq \frac{n'}{2}$. The other statements can be proven analogously.

Let us look at the graph of the proof of Theorem~\ref{circleconstruction} again, labeled as mentioned before. Choose some center edge $(n',i)$ where $i\leq \frac{n'}{2}$. Remember that $(n',i)$ corresponds to the length $n'$ binary vector with a $1$ in positions $i$ and $n'$ and zeros elsewhere. Hence the arising Ferrers diagram has only one row with exactly $(n'-i-1)$ many dots.

Now we look at all edges whose smaller entry $i'$ satisfies $1\leq i'<i$. Such an edge will always be of the form $(i-j,i+j)$ for $1\leq j < i$, thus there are $(i-1)$ of these edges. One can see by induction that all of these edges give rise to Ferrers diagrams of the same size, since a FD corresponding to $(x-1,y+1)$ can be obtained from the FD corresponding to $(x,y)$ by adding a point in the first row and deleting a point in the second row. We can count the dots e.g. in the FD corresponding to $(1,2i-1)$: There are $n'-2$ dots in the first row and $n'-(2i-1)$ in the second, hence a sum of $2n'-2i-1$ dots for the whole FD.

The edges that are left are of the form $(\frac{n'}{2}+i-1-j, \frac{n'}{2}+i+j)$ for $0\leq j< \frac{n'}{2}-i$.
With the same argument as in the paragraph before, all of these FD have the same number of dots and there are $\frac{n'}{2}-i$ many of them. We can count the dots in the FD arising from $(i, n'-1+i)$: There are $n'-1-i$ dots in the first row and $n'-(n'-1+i)$ in the second, hence a sum of $n'-2i$ dots for the whole FD.
%
\end{proof}



\subsection{Proof of the cardinalities in Theorems \ref{thm:c4} and \ref{thm:c5}}
\label{ap:B}

\begin{proof}
We derive the cardinalities of each component of the set of identifying vectors from Theorems~\ref{thm:c4} and~\ref{thm:c5}.

Let $n$ be even.
The FDRM code with rank distance $d =2$ arising from the identifying vectors of
\begin{itemize}
\item $\cA_1^4$ has cardinality $(q^{(n-4)}+ q^{(n-6)}) (q^{2(n-6)}+ (\frac{n}{2}-4)q^{(n-7)}+q^{(\frac{n}{2}-4)} ) $.
\item $\cA_2^4$ has cardinality $(q^{(n-5)}+ q^{(n-6)}) (\left(\frac{n}{2} - 4\right) q^{(n-8)} + q^{(2n-13)} + q^{(n-7)}) $.
\item $\cA_3^4$ has cardinality $(q^{(n-5)}+ q^{(n-6)})\times \\ \left[\sum_{i=2}^{\lceil\frac{q}{2}\rceil+1} (iq^{2n-2i-10} + (\frac{n-6}{2} -i) q^{n-2i-5} + q^{\frac{n-6}{2} -i})       + \right. \\ \left.      \sum_{i=2}^{\lfloor\frac{q}{2}\rfloor+1} (iq^{2n-2i-11} + (\frac{n-6}{2} -i) q^{n-2i-6} + q^{n -i-6})  \right]$.
\end{itemize}
\begin{itemize}
\item $\cA_1^5$ has cardinality $(q^{2n-10}+ q^{2n-14}) (q^{2(n-7)}+ (\frac{n-8}{2})q^{(n-9)}) $.
\item $\cA_2^5$ has cardinality $(q^{2n-11}+ q^{2n-13}) (\frac{n-8}{2}  q^{(n-10)} + q^{(2n-15)} ) $.
\item $\cA_3^5$ has cardinality $(q^{2n-12}+ q^{2n-13}) (2q^{2(n-8)}+ \frac{n-10}{2}q^{(n-11)} ) $.
\item $\cA_4^5$ has cardinality $(q^{2n-12}+ q^{2n-14}) (\frac{n-10}{2}q^{(n-12)} + 2q^{(2n-17)} ) $.
\item $\cA_5^5$ has cardinality $(q^{2n-12}+ q^{2n-14})\times \\ \left[\sum_{i=3}^{\lceil\frac{q}{2}\rceil+2} (iq^{2n-2i-12} + (\frac{n-6}{2} -i) q^{n-2i-7} )       + \right. \\ \left.      \sum_{i=3}^{\lfloor\frac{q}{2}\rfloor+2} (iq^{2n-2i-13} + (\frac{n-6}{2} -i) q^{n-2i-8} )  \right] $.
\end{itemize}

Let $n$ be odd.
The FDRM code with rank distance $d =2$ arising from the identifying vectors of
\begin{itemize}
\item $\cA_1^4$ has cardinality $(q^{(n-4)}+ q^{(n-6)}) (q^{2(n-6)}+ (\frac{n-3}{2})q^{(n-8)}) $.
\item $\cA_2^4$ has cardinality $(q^{(n-5)}+ q^{(n-6)}) (\frac{n-3}{2}  q^{n-9} + q^{2n-13}) $.
\item $\cA_3^4$ has cardinality $(q^{(n-5)}+ q^{(n-6)}) \times \\ \left[\sum_{i=2}^{\lceil\frac{q}{2}\rceil+1} (iq^{2n-2i-10} + (\frac{n-5}{2} -i) q^{n-2i-6} )       + \right. \\ \left.      \sum_{i=2}^{\lfloor\frac{q}{2}\rfloor+1} (iq^{2n-2i-11} + (\frac{n-5}{2} -i) q^{n-2i-7} )  \right]$.
\end{itemize}
\begin{itemize}
\item $\cA_1^5$ has cardinality $(q^{2n-10}+ q^{2n-14}) (q^{2n-14}+ (\frac{n-9}{2})q^{(n-8)} + q^{\frac{n-9}{2}}) $.
\item $\cA_2^5$ has cardinality $(q^{2n-11}+ q^{2n-13}) (\frac{n-9}{2}  q^{(n-9)} + q^{(2n-15)} + q^{n-8}) $.
\item $\cA_3^5$ has cardinality $(q^{2n-12}+ q^{2n-13}) (q^{2n-16}+ (\frac{n-11}{2})q^{(n-10)} + q^{\frac{n-11}{2}} ) $.
\item $\cA_4^5$ has cardinality $(q^{2n-12}+ q^{2n-14}) (\frac{n-11}{2}  q^{(n-11)} + 2q^{(2n-17)} + q^{n-9}   ) $.
\item $\cA_5^5$ has cardinality $(q^{2n-14}+ q^{2n-12})\times \\ \left[\sum_{i=3}^{\lceil\frac{q}{2}\rceil+2} (iq^{2n-12} + (\frac{n-7}{2} -i) q^{n-2i-6} + q^{\frac{n-7}{2} -i})       + \right. \\ \left.      \sum_{i=3}^{\lfloor\frac{q}{2}\rfloor+2} (iq^{2n-2i-13} + (\frac{n-7}{2} -i) q^{n-2i-7} + q^{n -i-7})  \right]$.
\end{itemize}

These formulas imply the cardinality formulas of  Theorems \ref{thm:c4} and \ref{thm:c5} by summing them up and adding the largest known code of length $n-k$ with zeros appended in front. Note that when summing them up we can merge the cardinalities of $\cA_2^4$ and $\cA_3^4$, as well as $\cA_4^5$ and $\cA_5^5$, respectively. An index shift in the second sums results in the formulas of Theorems~\ref{thm:c4} and~\ref{thm:c5}.
\end{proof}


\subsection{Proof of Lemma~\ref{lm:diff A-MC}}
\label{ap:C}

\begin{proof}
 Let $s=\sum_{i=3}^k i$ be as defined in Construction A and let $n':=n-s$.
 According to Construction A and by Theorem~\ref{trm:recursive_parameters} we have that the part of the code corresponding to the identifying vectors in $\cA_{0}^{k}\cup \{v_{00}^k\}$, denoted by $C_{A0}$,
is of cardinality $q^{2(n-k)}+q^{2(n-(k+(k-1)))}+\ldots+q^{2(n'+3)}+q^{2n'}$
and the cardinality of $C_A\setminus C_{A0}$  is $\Gauss{n'}{2}$.

The number of identifying vectors for $C_{A0}$ is $k-2$.
Let $N_s$ be the set of identifying vectors  for the multicomponent construction  with the first nonzero coordinate in the first $s=n-n'$ positions. Then $|N_s|=\lceil\frac{n-n'}{k-1}\rceil=\lceil\frac{k^2+k-6}{2(k-1)}\rceil\leq k-2$, for $k\geq 5$.
Hence, $|N_s|$ is at most the number of identifying vectors for $C_{A0}$.
 Denote the subcode of $C_{MC}$ corresponding to the identifying vectors in $N_{s}$ by $C_{MC}^{s}$. Then
$$ |C_{MC}^{s}| = q^{2(n-k)} + q^{2(n-2k+1)} + q^{2(n-3k+2)} + q^{2(n-4k+3)} + \dots$$
and
$$ |C_{A0}| = q^{2(n-k)} + q^{2(n-2k+1)} + q^{2(n-3k+3)} + q^{2(n-4k+6)}+ \dots$$
Since the number of summands in the former is at most the number of summands in the latter, we get that $ |C_{MC}^{s}| \leq  |C_{A0}| $.

Now we consider the identifying vectors of $C_{MC}^{n'}:=C_{MC}\backslash C_{MC}^{s}$, i.e.\ the identifying vectors with all nonzero entries contained in the last $n'$ coordinates. Let $x:=\left\lfloor\frac{n'-k}{k-1}\right\rfloor +1=\left\lfloor\frac{n'-1}{k-1}\right\rfloor$ be an upper bound on the number of identifying vectors in $C_{MC}^{n'}$. Then
\[|C_{MC}^{n'}|\leq
\sum_{i=0}^{x-1}q^{2(n'-k-i(k-1))}
=\sum_{i=1}^{x}q^{2(n'-1-i(k-1))}\]
\[=q^{2(n'-1)}\sum_{i=1}^{x}q^{-2i(k-1)}
=q^{2(n'-1)-2x(k-1)}\frac{(q^{2x(k-1)}-1)}{(q^{2(k-1)}-1)}\]
\[\leq
q^{2(n'-1)-2(k-1)+1}=q^{2n'-2k+1}.
\]
Then
\[|C_A\setminus C_{A0}|-|C_{MC}^{n'}|\geq \Gauss{n'}{2}-q^{2n'-2k+1}\]
\[\geq q^{2n'-4}-q^{2n'-2k+1}=q^{2n'-2k+1}(q^{2k-5}-1).
\]
Hence,
\[ |C_A|-|C_{MC}|\geq q^{2n'-2k+1}(q^{2k-5}-1)> q^{2n'-2k+1}q^{2k-6}\]
\[>q^{2n'-5}=q^{2n-k^2-k+1},
\]
and the statement of the lemma follows for $k\geq 5$.

For $k=4$, $\cA_{0}^{4}\cup \{v_{00}^4\}$ contains only two different identifying vectors which are identical to the first two vectors of $C_{MC}$.
Now we consider the identifying vectors of $C_{MC}$ with all nonzero entries contained in the last $n'+1$ coordinates. Note that with the first two identifying vectors these vectors are all the vectors of $C_{MC}$.
We denote by $C_{MC}^{n'+1}$ the part of the code $C_{MC}$ which corresponds to these vectors. Denote by $x:=\left\lfloor\frac{n'}{3}\right\rfloor$ an upper bound on the number of identifying vectors in $C_{MC}^{n'+1}$. Then
$$|C_{MC}^{n'+1}|\leq \sum_{i=1}^{x}q^{2(n'-3i)}=q^{2n'}\sum_{i=1}^{x}q^{-6i}$$
$$=q^{2n'-6x}\frac{(q^{6x}-1)}{(q^{6}-1)}\leq
q^{2n'-5}.
$$
Then
$$|C_A|-|C_{MC}|=|C_A\setminus C_{A0}|-|C_{MC}^{n'+1}|\geq \Gauss{n'}{2}-q^{2n'-5}$$
$$\geq q^{2n'-4}-q^{2n'-5}=q^{2n'-5}(q-1)\geq q^{2n'-5}=q^{2n-k^2-k+1}
$$
and the statement of the lemma follows for $k=4$.
\end{proof}

\subsection{Proof of Lemma~\ref{lm:diff B-MC}}
\label{ap:D}
\begin{proof}
By equation~(\ref{eq:MC}), the cardinality of the code $C_{MC}$ obtained by the multicomponent construction with distance $d=2$ is upper bounded by
$$|C_{MC}|\leq A^*_q(n-k,2,k)+\sum_{i=0}^{\lceil\frac{k}{2}\rceil-1}q^{(k-1)(n-k-2i)}
$$
$$= A^*_q(n-k,2,k)+q^{(k-1)(n-k)}+q^{(k-1)(n-k)}\sum_{i=1}^{\lceil\frac{k}{2}\rceil-1}q^{-2(k-1)i}
$$
$$=A^*_q(n-k,2,k)+q^{(k-1)(n-k)}+$$
$$q^{(k-1)(n-k)}\frac{q^{2(k-1)(\lceil\frac{k}{2}\rceil-1)}-1}{q^{2(k-1)(\lceil\frac{k}{2}\rceil-1)}(q^{2(k-1)}-1)}
$$
$$<A^*_q(n-k,2,k)+q^{(k-1)(n-k)}+q^{(k-1)(n-k-2)+1}$$
$$=A^*_q(n-k,2,k)+q^{(k-1)(n-k)}+q^{(k-1)(n-k)-2k+3}.
$$
On the other hand, by Theorem~\ref{thm:cardinalityB}, the cardinality of $C_B$ is bounded by
$$
|C_B|\geq
A^*_q(n-k,2,k) +q^{(k-1)(n-k)}+$$
$$q^{(k-1)(n-k)-4}+q^{(k-3)(n-k)-4+4\lfloor\frac{n-k}{2}\rfloor-4+2\epsilon(n-k)}+Z
$$
$$=
A^*_q(n-k,2,k)+q^{(k-1)(n-k)}+q^{(k-1)(n-k)-4}$$
$$+q^{(k-1)(n-k)-8}+Z,
$$
where
$$Z=\left(\sum_{i=2}^{\lfloor \frac{k-3}{2}\rfloor} q^{(k-3)(n-k)-4i} + \epsilon(k-1) q^{(k-3)(n-k-2)}\right)$$ $$\times\sum_{i=0}^{\lfloor\frac{n-k}{2}\rfloor -2} q^{2(2i+\epsilon(n-k))}>0.$$
Hence,
$$|C_B|-|C_{MC}|\geq Z+q^{(k-1)(n-k)-8}>q^{(k-1)(n-k)-8}.
$$
\end{proof}


\subsection{Proof of Lemma~\ref{lm:diff B-C}}
\label{ap:E}

\begin{proof}
First, by Theorem~\ref{thm:cardinalityB}, for the given parameters we have
$$|C_B|-2^{3(n-4)}-A^*_q(n-4,2,4)=(2^{n-4}+2^{n-6})\sum_{i=0}^{\frac{n-4}{2}-1}2^{4i}$$
$$=(2^{n-4}+2^{n-6})\frac{2^{2n-8}-1}{2^4-1}$$
$$\leq (2^{n-4}+2^{n-6})\frac{2^{2n-8}}{2^4-1} = \frac{2^{3n-12}+2^{3n-14}}{15}.$$
By Theorem~\ref{thm:c4} we have
$$|C_C|-2^{3(n-4)}-A^*_q(n-4,2,4)=$$
$$(2^{n-4}+2^{n-6})\left(2^{2n-12}+(\frac{n}{2}-4)2^{n-7}+2^{\frac{n}{2}-4}\right)+$$
$$
(2^{n-5}+2^{n-6})\Big(2^{2n-13}+\frac{n-10}{2}2^{n-9}+       2^{\frac{n-10}{2}}   +    2^{2n-14}$$
$$ +\frac{n-10}{2}2^{n-10}+2^{n-8}+
2^{2n-13}+\frac{n-8}{2}2^{n-8}+2^{n-7}\Big)
$$
$$=2^{3n-15}+2^{3n-20}+2^{3n-19}+X,$$
where
 $$X=(2^{n-4}+2^{n-6})\left((\frac{n}{2}-4)2^{n-7}+2^{\frac{n}{2}-4}\right)+$$
 $$(2^{n-5}+2^{n-6})\left(\frac{n-10}{2}(2^{n-9}+2^{n-10})
+2^{\frac{n-10}{2}} \right.$$
$$\left. +\frac{n-6}{2}2^{n-8}+2^{n-7}\right).$$
%


Now we find a lower bound on $X$. For $n\geq 10$ we have
$$X\geq \left(2^{n-4}+2^{n-6}\right)\left(2^{n-7}+2\right)$$
$$+\left(2^{n-5}+2^{n-6}\right)\left(2^{n-7}+2^{n-7}\right)
$$
$$\geq \left(2^{n-4}+2^{n-6}\right)2^{n-7}+\left(2^{n-5}+2^{n-6}\right)2^{n-6}$$
$$=2\cdot 2^{2n-11}+2^{2n-12}+2^{2n-13}\geq 2^{2n-10}.
$$

Since it holds that
$$2^{3n-15}\geq \frac{2^{3n-12}+2^{3n-14}}{15}   =  \frac{1}{3} 2^{3n-14}  ,
$$
it follows that
$$|C_C|-|C_B|\geq X+2^{3n-20}+2^{3n-19} \geq 3\cdot 2^{3n-20},$$
which implies the statement.
%
\end{proof}


\section*{Acknowledgment}

The authors wish to thank Tuvi Etzion  and Antonia Wachter-Zeh for many
helpful discussions.
They also thank Joachim Rosenthal for hosting the first author at the University of Zurich, where part of this work was done.

\bibliography{network_coding_stuff,network_coding_stuff_2}

\begin{thebibliography}{10}
\providecommand{\url}[1]{#1}
\csname url@samestyle\endcsname
\providecommand{\newblock}{\relax}
\providecommand{\bibinfo}[2]{#2}
\providecommand{\BIBentrySTDinterwordspacing}{\spaceskip=0pt\relax}
\providecommand{\BIBentryALTinterwordstretchfactor}{4}
\providecommand{\BIBentryALTinterwordspacing}{\spaceskip=\fontdimen2\font plus
\BIBentryALTinterwordstretchfactor\fontdimen3\font minus
  \fontdimen4\font\relax}
\providecommand{\BIBforeignlanguage}[2]{{%
\expandafter\ifx\csname l@#1\endcsname\relax
\typeout{** WARNING: IEEEtran.bst: No hyphenation pattern has been}%
\typeout{** loaded for the language `#1'. Using the pattern for}%
\typeout{** the default language instead.}%
\else
\language=\csname l@#1\endcsname
\fi
#2}}
\providecommand{\BIBdecl}{\relax}
\BIBdecl

\bibitem{ko08}
R.~K{\"o}tter and F.~R. Kschischang, ``Coding for errors and erasures in random
  network coding,'' \emph{IEEE Transactions on Information Theory}, vol.~54,
  no.~8, pp. 3579--3591, 2008.

\bibitem{bo09}
M.~Bossert and E.~Gabidulin, ``One family of algebraic codes for network
  coding,'' in \emph{Information Theory, 2009. ISIT 2009. IEEE International
  Symposium on}, 2009, pp. 2863--2866.

\bibitem{et09}
T.~Etzion and N.~Silberstein, ``Error-correcting codes in projective spaces via
  rank-metric codes and {F}errers diagrams,'' \emph{IEEE Transactions on
  Information Theory}, vol.~55, no.~7, pp. 2909--2919, March 2009.

\bibitem{et12}
------, ``Codes and designs related to lifted {MRD} codes,'' \emph{IEEE
  Transactions on Information Theory}, vol.~59, no.~2, pp. 1004 --1017, 2013.

\bibitem{et11}
\BIBentryALTinterwordspacing
T.~Etzion and A.~Vardy, ``Error-correcting codes in projective space,''
  \emph{IEEE Transactions on Information Theory}, vol.~57, no.~2, pp.
  1165--1173, 2011. [Online]. Available:
  \url{http://dx.doi.org/10.1109/TIT.2010.2095232}
\BIBentrySTDinterwordspacing

\bibitem{si11Enum}
N.~Silberstein and T.~Etzion, ``Enumerative coding for grassmannian space,''
  \emph{IEEE Transactions on Information Theory}, vol.~57, no.~1, pp. 365--374,
  2011.

\bibitem{si11Lexi}
------, ``Large constant dimension codes and lexicodes,'' \emph{Adv. in Math.
  of Comm.}, vol.~5, no.~2, pp. 177--189, 2011.

\bibitem{ga11}
E.~M. Gabidulin and N.~I. Pilipchuk, ``Multicomponent network coding,'' in
  \emph{Proceedings of the Seventh International Workshop on Coding and
  Cryptography (WCC) 2011}, Paris, France, 2011, pp. 443--452.

\bibitem{ga10}
M.~Gadouleau and Z.~Yan, ``Constant-rank codes and their connection to
  constant-dimension codes,'' \emph{IEEE Transactions on Information Theory},
  vol.~56, no.~7, pp. 3207--3216, 2010.

\bibitem{kh09}
A.~Khaleghi and F.~R. Kschischang, ``Projective space codes for the injection
  metric,'' \emph{arXiv:0904.0813v1 [cs.IT]}, 2009.

\bibitem{ko08p}
A.~Kohnert and S.~Kurz, ``Construction of large constant dimension codes with a
  prescribed minimum distance,'' in \emph{MMICS}, ser. Lecture Notes in
  Computer Science, J.~Calmet, W.~Geiselmann, and J.~M{\"u}ller-Quade, Eds.,
  vol. 5393.\hskip 1em plus 0.5em minus 0.4em\relax Springer, 2008, pp. 31--42.

\bibitem{ma08p}
F.~Manganiello, E.~Gorla, and J.~Rosenthal, ``Spread codes and spread decoding
  in network coding,'' in \emph{Proceedings of the 2008 IEEE International
  Symposium on Information Theory}, Toronto, Canada, 2008, pp. 851--855.

\bibitem{sk10}
V.~Skachek, ``Recursive code construction for random networks,''
  \emph{Information Theory, IEEE Transactions on}, vol.~56, no.~3, pp.
  1378--1382, 2010.

\bibitem{tr10}
A.-L. Trautmann and J.~Rosenthal, ``New improvements on the echelon-{F}errers
  construction,'' in \emph{Proceedings of the 19th International Symposium on
  Mathematical Theory of Networks and Systems -- MTNS}, Budapest, Hungary,
  2010, pp. 405--408.

\bibitem{kh09master}
A.~Khaleghi, ``Projective space codes for the injection metric,'' Master’s
  Thesis, University of Toronto, 2009.

\bibitem{kh09p}
A.~Khaleghi, D.~Silva, and F.~R. Kschischang, ``Subspace codes,'' in \emph{IMA
  Int. Conf.}, 2009, pp. 1--21.

\bibitem{si08j}
D.~Silva, F.~R. Kschischang, and R.~K\"otter, ``A rank-metric approach to error
  control in random network coding,'' \emph{IEEE Transactions on Information
  Theory}, vol.~54, no.~9, pp. 3951 --3967, 2008.

\bibitem{tr10p}
A.-L. Trautmann, F.~Manganiello, and J.~Rosenthal, ``Orbit codes - a new
  concept in the area of network coding,'' in \emph{IEEE Information Theory
  Workshop (ITW)}, Dublin, Ireland, 2010, pp. 1--4.

\bibitem{be75}
A.~Beutelspacher, ``Partial spreads in finite projective spaces and partial
  designs,'' \emph{Mathematische Zeitschrift}, vol. 145, no.~3, pp. 211--229,
  1975.

\bibitem{go13}
E.~Gorla and A.~Ravagnani, ``Partial spreads in random network coding,''
  \emph{Finite Fields and Their Applications}, vol.~26, pp. 104--115, 2014.

\bibitem{tr13phd}
A.-L. Trautmann, ``Constructions, decoding and automorphisms of subspace
  codes,'' Ph.D. dissertation, University of Zurich, Switzerland, 2013.

\bibitem{ga85a}
E.~M. Gabidulin, ``Theory of codes with maximum rank distance,'' \emph{Problemy
  Peredachi Informatsii}, vol.~21, no.~1, pp. 3--16, 1985.

\bibitem{ro91}
R.~Roth, ``Maximum-rank array codes and their application to crisscross error
  correction,'' \emph{IEEE Transactions on Information Theory}, vol.~37, no.~2,
  pp. 328 --336, mar 1991.

\bibitem{tu84}
W.~T. Tutte, \emph{Graph theory}, ser. Encyclopedia of Mathematics and its
  Applications.\hskip 1em plus 0.5em minus 0.4em\relax Reading, MA:
  Addison-Wesley Publishing Company Advanced Book Program, 1984, vol.~21, with
  a foreword by C. St. J. A. Nash-Williams.

\bibitem{li01b}
J.~H. {van}~Lint and R.~M. Wilson, \emph{A Course in Combinatorics},
  2nd~ed.\hskip 1em plus 0.5em minus 0.4em\relax Cambridge University Press,
  2001.

\bibitem{wa03a}
H.~Wang, C.~Xing, and R.~Safavi-Naini, ``Linear authentication codes: bounds
  and constructions,'' \emph{IEEE Transactions on Information Theory}, vol.~49,
  no.~4, pp. 866--872, April 2003.

\bibitem{xi09}
S.-T. Xia and F.-W. Fu, ``Johnson type bounds on constant dimension codes,''
  \emph{Designs, Codes and Cryptography}, vol.~50, no.~2, pp. 163--172, 2009.

\bibitem{wa13p}
A.~Wachter-Zeh and A.~Zeh, ``Interpolation-based decoding of interleaved
  gabidulin codes,'' in \emph{Preproceedings of the International Workshop on
  Coding and Cryptography (WCC) 2013}, Bergen, Norway, 2013, pp. 528--538.

\end{thebibliography}
\bibliographystyle{IEEEtran}

\end{document}